\documentclass[journal,final,twoside]{IEEEtranTCOM}
\usepackage{amsmath,amsfonts}
\usepackage{amsthm}
\usepackage{abraces}
\usepackage{algorithm}
\usepackage{algorithmic}
\usepackage{array}
\usepackage[caption=false]{subfig}
\usepackage{textcomp}
\usepackage{stfloats}
\usepackage{url}
\usepackage{verbatim}
\usepackage{graphicx}
\usepackage{cite}
\usepackage{wrapfig}
\usepackage{mathtools}
\usepackage{multirow}
\usepackage{color,soul}
\usepackage{amssymb}
\usepackage{setspace}
\usepackage{cases}
\usepackage{tikz}
\usepackage{optidef}
\usetikzlibrary{automata, positioning, arrows}
\usepackage{lipsum}
\usetikzlibrary{shapes.geometric}

\newtheorem{theorem}{Theorem}[]
\newtheorem{corollary}{Corollary}[]
\newtheorem{lemma}[]{Lemma}

\newtheorem{remark}{Remark}
\newcommand{\eqcite}[2]{[\citenum{#1}, eq. (#2)]}
\newcommand{\thref}[1]{\textbf{Theorem  \ref{#1}}}
\newcommand{\coref}[1]{\textbf{Corollary \ref{#1}}}
\newcommand{\leref}[1]{\textbf{Lemma \ref{#1}}}

\newcommand{\apref}[1]{A{\small PPENDIX} \ref{#1}}
\newcommand{\figref}[1]{Fig.~\ref{#1}}
\newcommand{\subfigref}[2]{Fig.~\ref{#1}\subref{#2}}
\newcommand{\secref}[1]{Section~\ref{#1}}
\newcommand{\tabref}[1]{{Table~\ref{#1}}}
\newcommand{\alref}[1]{{\textbf{Algorithm~\ref{#1}}}}

\DeclareMathOperator{\dv}{\boldsymbol{\mathnormal{d}}}
\DeclareMathOperator{\dr}{d}

\DeclareMathOperator{\pk}{peak}
\DeclareMathOperator{\nv}{\boldsymbol{\mathnormal{n}}}
\DeclareMathOperator{\I}{I}
\DeclareMathOperator{\G}{G}
\DeclareMathOperator{\CGM}{CGM}
\DeclareMathOperator{\parO}{\partial\Omega}
\DeclareMathOperator{\parOi}{\partial\Omega_\mathnormal{i}}
\DeclareMathOperator{\parOim}{\partial\Omega_\mathnormal{i}^\mathnormal{m}}
\DeclareMathOperator{\ai}{a_\mathnormal{i}}

\DeclareMathOperator{\rect}{rect}

\definecolor{ violet!10}{RGB}{1, 25, 89}
\definecolor{lime!10}{RGB}{104, 123, 62}
\definecolor{ violet!1043}{RGB}{157, 137, 43}
\definecolor{ orange!10}{RGB}{250, 204, 250}

\begin{document}
\raggedbottom
\interfootnotelinepenalty=10000
\title{Modeling and Optimization of Insulin Injection\\ for Type-1 Diabetes Mellitus Management} 
\author{Rinrada Jadsadaphongphaibool, Dadi Bi, Christian D. Lorenz, Yansha Deng,~\IEEEmembership{Senior Member,~IEEE}, and Robert Schober,~\IEEEmembership{Fellow,~IEEE}
\thanks{Rinrada Jadsadaphongphaibool, Dadi Bi, Christian D. Lorenz, and Yansha Deng are with the Department of Engineering, King’s College London, London WC2R 2LS, U.K. (e-mail:\{rinrada.jadsadaphongphaibool, dadi.bi, chris.lorenz, yansha.deng\}@kcl.ac.uk). (Corresponding author: Yansha Deng and Dadi Bi).}
\thanks{Robert Schober is with the Institute for Digital Communications, Friedrich Alexander-Universit\"{a}t (FAU) Erlangen-N\"{u}rnberg, 91058 Erlangen, Germany (e-mail: robert.schober@fau.de).}
\thanks{This work was supported by Engineering and Physical Sciences Research Council (EPSRC), U.K., under Grant EP/T000937/1.}
\thanks{This work was funded in part by the Deutsche Forschungsgemeinschaft (DFG, German Research Foundation) – GRK 2950 – Project-ID 509922606.}
\thanks{This paper has been accepted in part for the 2025 International Conference on Communications, June 2025.}
}

\markboth{}%
{}


\maketitle

\begin{abstract}
Diabetes mellitus is a global health crisis characterized by poor blood sugar regulation, impacting millions of people worldwide and leading to severe complications and mortality. Although Type 1 Diabetes Mellitus (T1DM) has a lower number of cases compared to other forms of diabetes, it is often diagnosed at a young age and requires lifelong exogenous insulin administration. In this paper, we focus on understanding the interaction of insulin and glucose molecules within the subcutaneous layer, which is crucial for blood sugar control in T1DM patients. Specifically, we propose a comprehensive model to characterize the insulin-glucose system within the subcutaneous layer, incorporating a multicellular molecular communication system. We then divide the T1DM system into insulin and glucose subsystems and derive the end-to-end expression for insulin-glucose interaction in the subcutaneous layer. We further validate the insulin-glucose interaction analysis with an agent-based simulator. As effectively managing postprandial glucose levels is crucial for individuals with T1DM to safeguard their overall health and avert short-term and long-term complications, we also derive the optimal insulin administration time based on the derived glucose response via the Lagrange multiplier and gradient descent ascent method. This allows us to explore the impact of different types of insulin and dietary management on blood sugar levels. Simulation results confirm the correctness of our proposed model and the effectiveness of our optimized effective time window for injecting insulin in individuals with T1DM.
\end{abstract}

\IEEEpeerreviewmaketitle
\section{Introduction}
Diabetes mellitus is a chronic medical condition characterized by elevated blood sugar levels due to disrupted regulation \cite{roglic2016global}. This disruption in glucose regulation can lead to severe complications, including cardiovascular issues, nerve damage, kidney failure, and other debilitating conditions, impacting both individual well-being and global healthcare systems. In 2021, it affected approximately 537 million people worldwide, leading to 6.7 million deaths in that year \cite{atlas2021international}. There are two types of diabetes: Type 1 Diabetes Mellitus (T1DM) which arises from an autoimmune attack on insulin-producing cells\cite{katsarou2017type}, and Type 2 Diabetes Mellitus (T2DM) which results from insulin resistance and reduced insulin production\cite{gonzalez2023evolutionary}. Although T1DM patients comprise only 8$\sim$10\% of the cases, the distinct nature of T1DM, the fact that it is often diagnosed in young people, and the need for lifelong insulin therapy underscore the need for a comprehensive understanding and dedicated research to manage this complex disease.

Various treatments and interventions are available for diabetes management to regulate blood sugar levels and mitigate associated complications \cite{suji2003approaches,chamberlain2016diagnosis}. Lifestyle modifications, including diet and exercise, are the initial and common recommendations in the early stages of T2DM, promoting healthy habits to manage blood sugar effectively \cite{colberg2010exercise,kirwan2017essential}. When T2DM progresses to the stage where lifestyle changes become ineffective, antidiabetic drugs, such as metformin and sulfonylureas, are frequently employed to improve insulin sensitivity and modulate glucose homeostasis by manipulating the metabolic system and its processes \cite{modi2007diabetes,wang2017efficacy}. However, lifestyle modifications and antidiabetic drugs are inadequate for T1DM as the body is unable to produce insulin and therapies aimed at modifying insulin sensitivity are ineffective. In such a case, the use of exogenous insulin by injection emerges as a crucial therapeutic choice for patients with T1DM \cite{mathieu2017insulin}, due to its ability to directly supplement deficient endogenous insulin production and maintain optimal blood sugar levels. Despite its effectiveness in diabetes therapy, excessive use of exogenous insulin can lead to insulin resistance and various side effects, e.g., weight gain, worsening of diabetic retinopathy, and breathing difficulties \cite{holman2009three,lebovitz2011insulin}. Thus, optimizing the timing of insulin injections to ensure administration at the most effective time and reduce the frequency of injections requires further study.

To optimize insulin injection for effective drug administration, the monitoring of blood sugar levels over time is essential. Continuous glucose monitor (CGM) systems have been developed to aid diabetic care by measuring the blood sugar level in real-time replacing the traditional finger-prick blood sample method\cite{vettoretti2020advanced}. When CGM is combined with an insulin pump, it can dynamically adjust insulin dosage concentration and timing to allow customized care tailored for the individual needs of patients\cite{tyler2020artificial}. To effectively adapt insulin injection, a mathematical model that can capture the physiological process of glucose and insulin interaction is needed.

The physiologically-based pharmacokinetic (PBPK) mathematical model was introduced in \cite{visentin2020padova} to investigate, observe, and predict the interaction of new drugs inside the human body. Its potential for optimization of meal size, timing, and insulin dosage for blood sugar control in T2DM patients when combined with an evolutionary algorithm was shown in\cite{gonzalez2023evolutionary}. PBPK exploits human physiology, population statistics, and drug characteristics to describe the pharmacokinetic interaction. However, it heavily relies on chemical property prediction models and may not account for physiological functions beyond the fluid circulation in patients \cite{khalil2011physiologically}. This can potentially lead to inaccuracies in the treatment outcome which results in less effective drug delivery.

To handle this limitation, multicellular modeling has been proposed to provide insights into the dynamic and complex interactions in living organisms by modeling the characteristics and behaviors of individual cells as well as their interactions with neighboring cells. Multicellular modeling has been used in plant biology \cite{bucksch2017morphological}, cancer biology \cite{venugopalan2014multicellular}, synthetic biology \cite{gorochowski2020toward}, and tissue engineering \cite{montes2019mathematical}. In multicellular modeling, communication between cells is facilitated by molecular communication (MC), which is a biological communication process that relies on chemical signaling \cite{bi2021survey}. Multicellular modeling effectively mimics the natural communication process of the considered biological system and allows optimization and customization of drug delivery treatment \cite{veiseh2010design}. While multicellular modeling has been explored in other contexts, further studies are required to fully understand its potential and optimize its application specifically for T1DM treatment.

Motivated by the above, the primary goal of this paper is to propose a foundational framework that characterizes the insulin-glucose system within the subcutaneous layer of T1DM patients. This framework can serve as the basis for future development, and its value is demonstrated by using it to optimize the timing of insulin injections. The contributions of this paper are outlined as follows: \begin{itemize}
    \item We propose a model for the insulin-glucose system in the subcutaneous layer of a T1DM patient. Our model framework is based on a multicellular MC system and consists of the insulin subsystem and the glucose subsystem capturing the spatial-temporal dynamics of insulin and glucose molecules. 
    \item We establish a mathematical framework for characterizing the T1DM system, which can be applied to analyze other multicellular environments in healthcare scenarios. Importantly, we apply Green's second identity to the partial differential equations (PDEs) to form the boundary integral equation (BIE)\cite{costabel1987principles} in \leref{th:1}. Unlike \cite{zoofaghari2021semi}, in \thref{th:BIEsolution}, we use approximations and the Adomian decomposition methods\cite{rahman2007integral} to derive an analytical expression for the concentration of insulin and glucose molecules in \textbf{Corollaries 1-3}. To validate our derivations, we utilize an agent-based simulator to characterize the cellular behavior of our proposed model. The results obtained from the simulator validate the proposed insulin-glucose interaction model and our derived mathematical framework.
    \item We investigate and determine the most effective time window for injecting insulin for patients with T1DM to manage postprandial glucose levels (the amount of glucose in the blood after a meal) in \textbf{Theorems 4} and \textbf{5}. We first consider a specific case that simplifies the environmental complexity, leading to the derivation of a closed-form expression for an upper limit on the insulin injection time. We also numerically determine an upper limit for the insulin injection time when the environmental conditions are more complex.   
    \item We investigate the impact of various types of insulin on the blood sugar level. Additionally, we examine the effect of dietary management on the ideal insulin injection time window. Our results reveal that carefully selecting the insulin type and effective dietary management can significantly enhance the quality of life of patients.
\end{itemize}

The remainder of this paper is organized as follows. In \secref{sec:T1DM}, we propose a model to mathematically describe a T1DM system in the subcutaneous layer. In \secref{sec:recob}, we apply the BIE to characterize our proposed T1DM model. In \secref{sec:opt}, we propose an optimization problem to obtain the optimal time for insulin injection after a meal. Numerical results in \secref{sec:numres} validate the proposed system model. Finally, we conclude this paper in \secref{sec:concl}. The main notations used in this paper are summarized in Table I.

\begin{table*}[!t]
    \centering
    \caption{Table of Notations}
    \label{tab:notation} 
    \begin{tabular}{p{0.15\linewidth}|c|p{0.65\linewidth}}
    \hline
    Notation&Unit&Definition\\\hline\hline
         $\Omega$&&Continuous interstitial fluid\\
         $N_{b}$, $N_{b}^R$, $N_{b}^D$&&Total number of disjoint boundaries\\
         $N_c$&&Total number of cells\\
         $\parOi$&&Set of all the center points on the $i$th boundary\\
         $\dv_{0}$, $\dv_{0_{\I}}$, $\dv_{0_{\G}}$&$\mu$m&Initial position where the molecules is released\\
         $t_{0}$, $t_{0_{\I}}$, $t_{0_{\G}}$&h&Initial time when the molecules are released\\
         $F_{\I}(\cdot)$, $F_{\G}(\cdot)$&nM&Input concentration\\
         $N_{\I_{0}}$&&Number of input insulin molecules\\
         $V_{\I_{0}}$&$\mu\mathrm{m}^3$&Volume of input insulin molecules\\
         $\dv$, $\dv'$, $\dv''$&$\mu$m&Vector defining position in space\\
         $c_{\I}(\cdot)$&nM&Insulin concentration\\
         $D$, $D_{\I}$, $D_{\G}$&$\mu$m$^2$ s$^{-1}$&Diffusion coefficient\\
         $k_{\dr}$, $k_{\dr_{\I}}$, $k_{\dr_{\G}}$&$\mathrm{h}^{-1}$&Degradation rate\\
         $\nv$, $\nv_{N_b}$, $\nv_{i}$&&Unit normal vector pointing outward from the surface\\
         $k_{\mathrm{a}}$, $k_{\mathrm{a}_{N_b}}$, $k_{\ai}$&$\mu$m s$^{-1}$&The absorption rate constant\\    
         $c_{\G}(\cdot)$&nM&Glucose concentration\\
         $\tau_{\G_{2}}$&h&Time taken for GLUT4 translocation\\
         $\tau_{\G_{3}}$&h&Duration of glucose influx\\
         $\overline{N_{\I}}(\cdot)$&&Average number of insulin molecules absorbed across all cells\\
         $\tau_{\G_{1}}$&h&Time taken after insulin injection for most cells to have received insulin\\
         $h(\cdot)$&nM&Molecular concentration in an unbounded environment\\
         $M_{i}$&&Number of sub-boundaries or meshes\\
         $S^m_{i}$&$\mu$m$^2$&Surface area of the mesh\\
         $\dv^m_{i}$&$\mu$m&Center of the mesh\\
         $\tau_{0_{\I}}$&h&Ideal window for insulin injection\\
         $\hat{t}_{0_{\I}}$&h&Upper bound on the insulin injection time after a meal\\
         $\phi$&nM&Post-meal glucose level threshold\\
         $c_{\pk}$&nM&Peak glucose concentration\\
         $\gamma_1$, $\gamma_2$&&Learning rate\\
\hline
\end{tabular}
\end{table*}

\section{System Model\label{sec:T1DM}}


To evaluate the effect of insulin intake on blood sugar control for patients with T1DM, we consider the system shown in \figref{fig:systemModel}(a), particularly focusing on the subcutaneous layer, located beneath the skin. This layer is commonly used for glucose monitoring and insulin injection due to its accessibility and rich interstitial fluid environment. Thus, our system comprises the injection of exogenous insulin into the subcutaneous layer while using a CGM to monitor glucose levels. Furthermore, we characterize the resulting distribution of insulin and glucose molecules within the subcutaneous layer to analyze the interactions between insulin and glucose.

\begin{figure}[h]
    \centering
    \includegraphics[width=\linewidth]{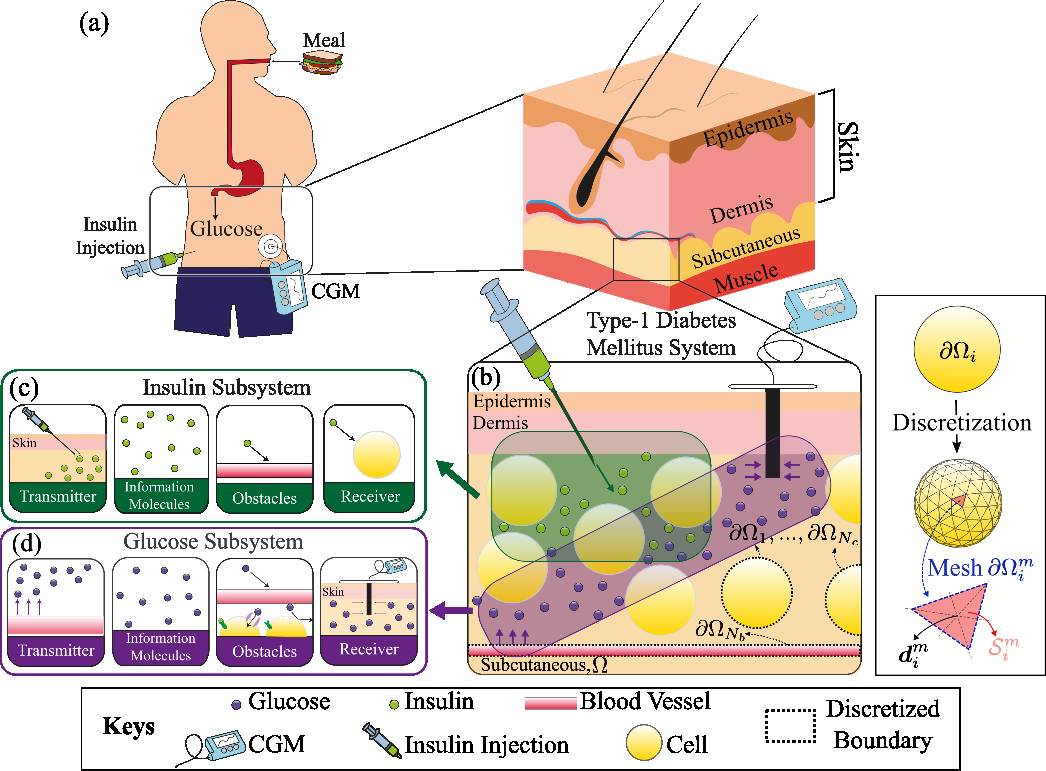}
    \caption{System Model. (a) Illustration of a T1DM patient. (b) Cross-section of T1DM patient's subcutaneous layer. (c) Insulin subsystem. (d) Glucose subsystem.}
    \label{fig:systemModel}
\end{figure}

The subcutaneous layer can be modeled as a 3-dimensional (3D) diffusion-based MC channel constituted by the continuous interstitial fluid environment $\Omega$, as shown in Fig. 1(b). We assume that the subcutaneous layer comprises $N_b$ disjoint boundaries, including $N_c$ cells and a blood vessel. Each boundary is discretized into $M_{i}$ meshes, where the $m$th mesh of the $i$th boundary, i.e., $\parOim$, is represented by a center point $\dv_{i}^{m}$ with a surface area $S_i^m$. The set of all the center points on the $i$th boundary is denoted as $\partial \Omega_i$. We also assume that the first $N_c$ boundaries are all cell membranes and the $N_b$th boundary is the blood vessel wall. The mesh center points on all the cell membrane boundaries and the blood vessel boundary are denoted as $\parO_{c}=\bigcup\limits_{i=1}^{N_c}\parOi$ and $\parO_{N_b}$, respectively. For the ease the analysis of the insulin-glucose system, we divide the considered system into two subsystems (see \figref{fig:systemModel}): the insulin subsystem and the glucose subsystem, which are described in detail in Sections~\ref{sec:insulin} and \ref{sec:glucose}, respectively.

\subsection{Insulin Subsystem\label{sec:insulin}}
As shown in \figref{fig:systemModel}(c), the insulin subsystem comprises the insulin release, where a syringe is used as the transmitter, insulin propagation with the blood vessel as an obstacle, and insulin reception, where the cells are the receivers. We mathematically model this subsystem in the following manner:
\subsubsection{Insulin Injection}
Insulin, injected just beneath the skin into the subcutaneous layer, is vital to manage diabetes by regulating blood sugar levels. We assume that insulin is injected in an impulsive manner into the subcutaneous layer at location $\dv_{0_{\I}}$ and at time $t_{0_{\I}}$, which can be modeled as follows \cite{Jamali2019Channel}\begin{align}
   F_{\I}(\dv_{0_{\I}},t_{0_{\I}}) = \dfrac{N_{\I_{0}}}{V_{\I_{0}}}\delta(\dv-\dv_{0_{\I}})\delta(t-t_{0_{\I}}), \label{eq:insulinIC}
\end{align} where $F_{\I}(\dv_{0_{\I}},t_{0_{\I}})$ is the input concentration of the insulin, $N_{\I_{0}}$ is the number of input insulin molecules, $V_{\I_{0}}$ is the volume of the input insulin molecules, $\dv=[x,y,z]$ is a vector defining a position in space, and $\delta(\cdot)$ is the Dirac delta function.

\subsubsection{Insulin Propagation}
After injection, insulin molecules propagate inside the subcutaneous layer, and the concentration of insulin in the subcutaneous layer can be described by the diffusion-reaction equation as follows \eqcite{zoofaghari2021semi}{5}\begin{align}
     \dfrac{\partial c_{\I}(\dv,t|\dv_{0_{\I}},t_{0_{\I}})}{\partial t} =& D_{\I}\nabla_{\dv}^{2}c_{\I}(\dv,t|\dv_{0_{\I}},t_{0_{\I}})\nonumber\\&-k_{\dr_{\I}}c_{\I}(\dv,t|\dv_{0_{\I}},t_{0_{\I}}) + F_{\I}(\dv_{0_{\I}},t_{0_{\I}}),\label{eq:cIPDE}
\end{align} where $c_{\I}$ is the concentration of insulin, $D_{\I}$ is the diffusion coefficient of insulin, $\nabla_{\dv}^2$ is the Laplace operator, and $k_{\dr_{\I}}$ is the degradation rate of insulin. During propagation, the molecules can either be absorbed by the cells or encounter the blood vessel. The blood vessel absorbs insulin molecules from the venous side and circulates them throughout the body. Over time, these molecules can be re-released into the subcutaneous layer through filtration at the arterial end \cite{scanlon2018essentials}. Since the propagation within the blood vessel is not the primary focus of this study, we model the blood vessel as a partially absorbing boundary. As such, the corresponding boundary condition is given by \eqcite{Deng2016MolecularReceiver}{5} \begin{align}
    D_{\I}\nabla_{\dv} c_{\I}(\dv,t|\dv_{0_{\I}},t_{0_{\I}})\cdot\nv_{N_b}|_{\dv\in \parO_{N_b}} = k_{\mathrm{a}_{N_b}}c_{\I}(\dv,t|\dv_{0_{\I}},t_{0_{\I}}), \label{eq:BC}
\end{align} where $\nabla_{\dv}$ is the Nabla operator, operator $\cdot$ denotes the inner product, $\nv_{N_b}$ is the unit normal vector pointing outward from the surface of the blood vessel, and $k_{\mathrm{a}_{N_b}}$ is the absorption rate constant of the blood vessel.


\subsubsection{Insulin Reception}
As the reception of insulin by cells is based on the receptors distributed on the cell surface, each cell can be considered as an active receiver, which can be modeled as follows \eqcite{Deng2016MolecularReceiver}{5}\begin{align}
        D_{\I}\nabla_{\dv} c_{\I}(\dv,t|\dv_{0_{\I}},t_{0_{\I}})\cdot\nv_{i}|_{\dv\in \parOi} = k_{\ai}c_{\I}(\dv,t|\dv_{0_{\I}},t_{0_{\I}}). \label{eq:BCins}
\end{align} where $\nv_i$ is the unit normal vector pointing outward from the surface at point $\dv\in \parOi$, and $k_{\ai}$ is the absorption rate constant of the boundary $\parOi$.

\subsection{Glucose Subsystem\label{sec:glucose}}
\begin{figure}
    \centering
    \subfloat[Cell reception of glucose molecules without insulin \label{fig:noinsulin}]{\includegraphics[width=0.45\linewidth]{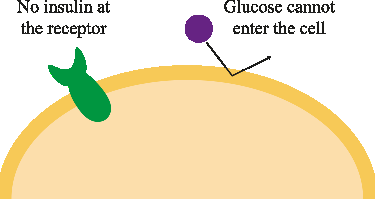}} \hfill
    \subfloat[Cell reception of glucose molecules with insulin \label{fig:withinsulin}]{\includegraphics[width=0.45\linewidth]{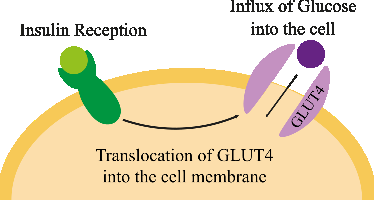}}\\
    \caption{The insulin-glucose reception mechanism by a cell.}
    \label{fig:CellRep}
\end{figure}

As shown in \figref{fig:systemModel}(d), the glucose subsystem comprises the glucose from the meal released by the blood vessel which serves as transmitter, glucose propagation inside the subcutaneous layer with cells and the blood vessel as obstacles, and glucose reception by the CGM. We mathematically model this subsystem as follows.
\subsubsection{Glucose Emission}
Glucose is derived from the breakdown of carbohydrates in the meal and is absorbed into the bloodstream to provide energy for the body's needs. We consider an impulsive release of glucose into the subcutaneous layer at location $\dv_{0_{\G}}$ and time $t_{0_{\G}}$ for simplicity, as the response of a system to any arbitrary input can be calculated based on the system's impulse response (i.e., the system's response to an impulse input) using a convolution integral. The impulsive input is modeled as \begin{align}
    F_{\G}(\dv_{0_{\G}},t_{0_{\G}}) = \dfrac{N_{\G_{0}}}{V_{\G_{0}}}\delta(\dv-\dv_{0_{\G}})\delta(t-t_{0_{\G}}),\label{eq:glucoseIC}
\end{align} where $N_{\G_{0}}$ is the number of input glucose molecules, $V_{\G_{0}}$ is the volume of the input glucose molecules.

\subsubsection{Glucose Propagation}
Unlike insulin, for glucose, we assume that changes in concentration inside the body are mainly caused by cell reception, which implies that glucose degradation is negligible. Thus, the glucose concentration in the subcutaneous layer can be described as \eqcite{zoofaghari2021semi}{5} \begin{align}
     \dfrac{\partial c_{\G}(\dv,t|\dv_{0_{\G}},t_{0_{\G}})}{\partial t} &= D_{\G}\nabla_{\dv}^{2}c_{\G}(\dv,t|\dv_{0_{\G}},t_{0_{\G}}) \nonumber\\&+ F_{\G}(\dv_{0_{\G}},t_{0_{\G}}),\label{eq:cG}
\end{align} where $D_{\G}$ is the diffusion constant of glucose. During propagation, glucose uptake by a cell occurs exclusively when insulin is present; in the absence of insulin, glucose cannot permeate through the cell membrane\cite{petersen2018mechanisms}. Insulin reception initiates a cascade of protein activations, e.g., the translocation of the glucose transporter type 4 (GLUT4) into the cell membrane. This transporter facilitates glucose influx into the cell, enabling glucose metabolism for energy production. While our model neglects these intracellular processes for simplicity, we assume the GLUT4 translocation into the cell membrane takes $\tau_{\G_2}$ seconds, and the GLUT4 transporter remains functional on the cell membrane for $\tau_{\G_3}$ seconds during which glucose influx can occur. Accordingly, the reception of the glucose molecules by a cell falls into states $S_0$ and $S_1$:
\begin{itemize}
    \item \textbf{Glucose Rejection:} State $S_0$ implies that there is no or negligible insulin reception by a cell. When cells are not capable of absorbing glucose molecules (see \subfigref{fig:CellRep}{fig:noinsulin}), each cell membrane can be modeled as a reflective boundary \eqcite{hahn2012heat}{1-56}, i.e., \begin{align}
    D_{\G}\nabla_{\dv} c_{\G}(\dv,t|\dv_{0_{\I}},t_{0_{\I}})\cdot\nv_{i}|_{\dv\in \parOi} = 0. \label{eq:BC-reflective}
    \end{align} 
    \item \textbf{Glucose Influx:} State $S_1$ implies that insulin molecules have arrived at the cell, which triggers the reception of glucose. When an influx of glucose molecules into a cell occurs (see \subfigref{fig:CellRep}{fig:withinsulin}), the cell membrane can be modeled as a fully absorbing boundary, i.e., $k_{\ai} = \infty$. For a fully absorbing receiver, it can be assumed that any glucose molecule that reaches the receiver is absorbed as soon as it hits the surface; thus, the corresponding boundary condition can be modeled as a Dirichlet boundary condition as follows \eqcite{hahn2012heat}{1-54}\begin{align}
    c_{\G}(\dv,t|\dv_{0_{\I}},t_{0_{\I}})|_{\dv\in \parOi} = 0. \label{eq:BC-absorbing}
    \end{align}
\end{itemize}

In practice, each cell independently transitions from state $S_0$ to $S_1$. We simplify the glucose system by assuming that all cells change their states simultaneously to make the analysis tractable, i.e., when the average number of insulin molecules absorbed across all cells $\overline{N_{\I}}(\dv\in\parO_{c},t|\dv_{0_{\I}},t_{0_{\I}})=1$, all cells transition from state $S_0$ to $S_1$. \footnote{We have verified via simulations that the error caused by this assumption is negligible, see \secref{sec:numres}.} We define $\tau_{\G_1}$ as the time it takes after insulin injection at $\dv_{0_{\I}}$ for most cells to have received insulin.

Apart from the reception by cells, the blood vessel can also take in glucose molecules. The corresponding boundary condition for the blood vessel is similar to \eqref{eq:BC}.


\subsubsection{Glucose Reception by CGM}
The glucose molecules propagating inside the subcutaneous layer can be detected by a CGM device that occupies space $\Omega_{\CGM}$. CGM is a commonly used wearable technology that constantly tracks glucose levels throughout the day, offering a comprehensive view of fluctuations that traditional finger-prick tests would otherwise miss\cite{suji2003approaches}. It consists of a small needle placed under the skin that measures the blood sugar levels in the interstitial fluid through a chemical reaction with the glucose molecules that produce detectable electrical signals. For simplicity, the CGM is modeled as a volume passive receiver. This facilitates the analysis of the glucose distribution and concentration measurement, avoiding the intricacies of the underlying electrochemical processes.


\section{Receiver Observation\label{sec:recob}}
\begin{figure}
    \centering
    \tikzstyle{startstop} = [rectangle, rounded corners, minimum width=3cm, minimum height=1cm,text centered, draw=black,fill=red!5]
    \tikzstyle{io} = [trapezium, trapezium left angle=70, trapezium right angle=110,  minimum width=3cm, minimum height=1cm, text centered, text width=1.5cm, draw=black,fill=blue!10]
    \tikzstyle{io2} = [trapezium, trapezium left angle=70, trapezium right angle=110,  minimum width=3cm, minimum height=1cm, text centered, draw=black,fill=blue!10]
    \tikzstyle{io3} = [trapezium, trapezium left angle=70, trapezium right angle=110,  minimum width=3cm, minimum height=1cm, text centered, text width=3cm, draw=black,fill=blue!10]
    \tikzstyle{process} = [rectangle, minimum width=3cm, minimum height=1cm, text centered, text width=5cm, draw=black]
    \tikzstyle{process2} = [rectangle, minimum width=3cm, minimum height=1cm, text centered, draw=black]

    \begin{tikzpicture}[scale=0.75,node distance=1.6cm,every node/.style={scale=0.75}]
        \node (in1) [io] {PDE with boundary conditions};
        \node (in2) [io, right of=in1, xshift=2cm] {Fundamental PDE};
        \node (pro1) [process, below of=in1] {Green's Second Identity};

        \node (in3) [io2, below of=pro1] {BIE};
        \node (pro2) [process, below of =in3]{Adomian Decomposition};
        \node (in4) [io3, below of=pro2] {Solution to the PDE with boundary condition};

        \node (pro3) [process2, below of =in2, yshift=-1.6cm]{Solution};


        \draw [->] (in1) -- (pro1);
        \draw [->] (pro1) -- (in3);
        \draw [->] (in3) -- (pro2);
        \draw [->] (pro2) -- (in4);
        \draw [->] (in2) |- (pro1);
        \draw [->] (in2) -- (pro3);
        \draw [->] (pro3) |- (pro2);
    \end{tikzpicture}
    \caption{Methodology flowchart for solving PDEs. Blue parallelograms represent equations and white blocks represent processes.}
    \label{fig:BIEsys}
\end{figure}

In this section, we solve the PDEs in \eqref{eq:cIPDE} and \eqref{eq:cG}, which are subject to multiple boundary conditions that are difficult to handle with conventional methods (e.g., separation of variables, transform method, and method of characteristics\cite{PDEforScientist}), by turning them into a boundary integral equation (BIE). As shown in \figref{fig:BIEsys}, Green's second identity can be applied to PDEs to reformulate a problem into a BIE\cite{costabel1987principles,liu2009fast}. This reduces the problem's dimensionality and naturally incorporates the boundary conditions into the equation, making it useful for problems where interior behavior (e.g., the mechanism inside the cell) is less important. Then, through the application of the Adomian decomposition method, the solution to the reformulated problem can be found. We present the new results in the form of lemmas, corollaries, and theorems in order to distinguish them from the known mathematical formulations used as a starting point for our analysis.

In the following, we first present the derivation of the BIE for the generic case of a PDE with Dirichlet's and Robin's boundary condition in \secref{sec:genBIE}. We then analyze the insulin subsystem using the BIE derived with the Adomian decomposition method in \secref{sec:insulinSA}. Based on the same principle, we also analyze the glucose subsystem in \secref{sec:glucoseSA}.

\subsection{General Boundary Integral Equation}\label{sec:genBIE}
As shown in \figref{fig:BIEsys}, explicit knowledge of the solution of the fundamental differential equation (i.e., the differential equation without the boundary conditions) is required to form the BIE. The fundamental differential equation can be formulated as an unbounded diffusion-reaction equation with an impulsive input at an arbitrary location $\dv'$ and at time $t=0$ as follows\begin{align}
    \dfrac{\partial h(\dv,t|\dv')}{\partial t} =& D\nabla_{\dv}^{2}h(\dv,t|\dv')-k_{\dr}h(\dv,t|\dv')\nonumber\\& + \delta(\dv-\dv')\delta(t), \label{eq:h}
\end{align}where $h(\dv,t|\dv')$ is the molecular concentration in an unbounded environment. To determine the generic BIE that fits both the insulin and glucose subsystems, we combine the two types of boundary conditions in \eqref{eq:BC} and \eqref{eq:BC-absorbing} as follows\begin{align}
    \begin{cases}
        D\nabla_{\dv} c(\dv,t|\dv_{0},t_{0})\cdot\nv_i|_{\dv\in \parOi} = k_{\ai}c(\dv,t|\dv_{0},t_{0}),\\
        \hfill i = 1,\cdots,N_b^R,\vspace{2ex}\\
        c(\dv,t|\dv_{0},t_{0})|_{\dv\in \parOi} = 0,\\
        \hfill i=N_b^R + 1,\cdots,N_b^R + N_b^D,
    \end{cases},\label{eq:generalBC}
\end{align} where $N_b^R$ and $N_b^D$ denote the number of boundaries with generalized Robin and Dirichlet boundary conditions, respectively. Relying on Green's second identity \cite{sadiku2000numerical}, we can reformulate the diffusion-reaction equation in \eqref{eq:cIPDE} in terms of the fundamental equation in \eqref{eq:h} and the boundary conditions in \eqref{eq:generalBC}.
\begin{lemma}\label{th:1}
    The BIE for the diffusion-reaction equation in \eqref{eq:cIPDE} with the boundary conditions in \eqref{eq:generalBC} can be derived as \begin{align}
            &C(\dv',\omega|\dv_{0},t_{0})= \nonumber\\ &\sum_{i=N_b^R + 1}^{N_b}\oint_{\parOi} H(\dv,\omega|\dv') \nabla_{\dv} C(\dv,\omega|\dv_{0},t_{0})\cdot\nv_i\mathrm{d}\mathcal{S}\nonumber\\&+\sum_{i=1}^{N_b^R}\oint_{\parOi} C(\dv,\omega|\dv_{0},t_{0})[k_{\ai}H(\dv,\omega|\dv')\nonumber\\ &-D\nabla_{\dv} H(\dv,\omega|\dv')\cdot\nv_i] \mathrm{d}\mathcal{S} + H(\dv_{0},\omega|\dv')e^{-j\omega t_{0}},\label{eq:generalBIE}
        \end{align}where $C(\dv,\omega|\dv_{0},t_{0})$ and $H(\dv,\omega|\dv')$ are the Fourier transforms of $c(\dv,t|\dv_{0},t_{0})$ and the solution to the fundamental solution in \eqref{eq:h}, respectively. In \eqref{eq:generalBIE}, $N_b = N_b^R + N_b^D$ and $\oint_{\parOi}(\cdot)\mathrm{d}\mathcal{S}$ represents the surface integral over the closed boundary of the $i$th object.
\end{lemma}
\begin{proof}
   Please refer to \cite{zoofaghari2021semi} for a comprehensive explanation of the derivation.
\end{proof}
Eq.~\eqref{eq:generalBIE} is in the form of a Fredholm integral equation of the second kind which can be solved using numerical techniques (e.g., the boundary element method (BEM) \cite{liu2009fast}) to determine the unknown concentration $C(\dv,\omega|\dv_{0},t_{0})$. However, the presence of multiple objects in the environment necessitates the creation of separate linear equations for each boundary, resulting potentially in a high computational complexity. The Adomian decomposition method is a systematic and versatile approach for solving nonlinear problems, notable for its grid-free methodology that streamlines computation and implementation compared to other numerical methods \cite{rahman2007integral}. Thus, we adopt the Adomian decomposition method to analyze \eqref{eq:generalBIE}. Then, to compute the surface integrals, each boundary $\parOi$ is discretized into $M_i$ sub-boundaries called meshes $\parOim$, with the area and center of each mesh denoted by $\mathcal{S}_i^m$ and $\dv_i^m$, respectively. The solution to the considered BIE is provided in the following theorem.
\begin{theorem}\label{th:BIEsolution}
    The solution to the BIE in \eqref{eq:generalBIE} is given by\begin{align}
        c(\dv,t|\dv_{0},t_{0}) &=  \sum_{n=0}^{\infty} c_n(\dv,t|\dv_{0},t_{0}),\label{eq:BIEsolution}
    \end{align} where \begin{align}
        c_0(\dv,t|\dv_{0},t_{0}) &= h(\dv_{0},t-t_{0}|\dv), \label{eq:c0general}
    \end{align} \begin{align}
        c_n(\dv,t&|\dv_{0},t_{0}) = \sum_{i=1}^{N_b^R}\sum_{m=1}^{M_i}c_{n-1}(\dv_i^m,t|\dv_{0},t_{0})\nonumber\\&* \mathcal{S}_i^m[k_{\ai}h(\dv_i^m,t|\dv)-D\nabla_{\dv} h(\dv_i^m,t|\dv)\cdot\nv_i^m]\nonumber\\&+\sum_{i=N_b^R+1}^{N_b}\sum_{m=1}^{M_i} \nabla_{\dv} c_{n-1}(\dv_i^m,t|\dv_{0},t_{0})\cdot\nv_i^m\nonumber\\&* D\mathcal{S}_i^m h(\dv_i^m,t|\dv),\label{eq:cngeneral}
    \end{align}and $*$ is the convolution operator.
\end{theorem} 
\begin{proof}
   See \apref{ap:2}.
\end{proof}
\begin{remark}
    Eqs. (12)-(14) reveal that the concentration comprises a term corresponding to an unbounded environment, i.e., $c_{0}(\dv,t|\dv_0,t_0)$ in (13), and adjustment terms accounting for boundary effects, i.e., $c_{n}(\dv,t|\dv_0,t_0)$ in (14). As will be shown in the following subsections, the application of this theorem extends naturally to \textbf{Corollaries 1}, \textbf{2}, and \textbf{3} and demonstrates the theorem's generality and versatility in addressing a broader range of scenarios within the considered system.
\end{remark}

\subsection{Insulin Subsystem Analysis\label{sec:insulinSA}}
The insulin subsystem is governed by the PDE in \eqref{eq:cIPDE} with the boundary conditions specified in \eqref{eq:BC} and \eqref{eq:BCins}. Based on \thref{th:1}, we present the insulin concentration in the following corollary.
\begin{corollary} \label{co:cI}
    The insulin concentration is given by \eqref{eq:BIEsolution}, where $h(\dv,t|\dv')$ is the solution of the diffusion-reaction equation in \eqref{eq:h} for the unbounded environment and is given by \eqcite{Jamali2019Channel}{30} \begin{align}
    h(\dv,t|\dv') = \dfrac{1}{(4\pi D_{\I}t)^{3/2}}\exp{\bigg( -\dfrac{||\dv-\dv'||^2}{4D_{\I}t}-k_{\dr_{_{\I}}}t\bigg)},\label{eq:hsol}
\end{align} and $c_0(\dv,t|\dv_{0_{\I}},t_{0_{\I}})$ and $c_n(\dv,t|\dv_{0_{\I}},t_{0_{\I}})$ can be respectively derived as \begin{align}
        &c_0(\dv,t|\dv_{0_{\I}},t_{0_{\I}}) =\nonumber\\& \dfrac{1}{\big(4\pi D_{\I}(t-t_{0_{\I}})\big)^{3/2}}\exp{\bigg( -\dfrac{||\dv_{0_{\I}}-\dv||^2}{4D_{\I}(t-t_{0_{\I}})}-k_{\dr_{\I}}(t-t_{0_{\I}})\bigg)}, \label{eq:c0ins-2}
    \end{align} and
    \begin{align}
       &c_n(\dv,t|\dv_{0_{\I}},t_{0_{\I}}) = \sum_{i=1}^{N_b}\sum_{m=1}^{M_i} c_{n-1}(\dv_i^m,t|\dv_{0_{\I}},t_{0_{\I}})* \dfrac{\mathcal{S}_i^m}{(4\pi D_{\I}t)^{3/2}}\nonumber\\&\times\bigg\{k_{\ai}+D_{\I} \dfrac{(\dv_i^m-\dv)\cdot\nv_i^m}{4D_{\I}t}\bigg\}\exp{\bigg( -\dfrac{||\dv_i^m-\dv||^2}{4D_{\I}t}-k_{\dr_{\I}}t\bigg)}.\label{eq:cnins-2}
    \end{align} 
\end{corollary}


As the glucose subsystem is dependent on the reception of insulin by a cell, we characterize the average number of insulin molecules absorbed by cell membrane $\parOi$ in the following theorem using the analysis in \coref{co:cI}.
\begin{theorem}
    The average number of insulin molecules received by the $N_{c}$ cells is given by\begin{align}
    \overline{N_{\I}}(\dv&\in\parO_{c},t|\dv_{0_{\I}},t_{0_{\I}})=\nonumber\\ & \dfrac{1}{N_{c}}\sum_{i=1}^{N_{c}}\sum_{m=1}^{M_i} \int_{t-\tau_{\G_2}}^{t-\tau_{\G_3}}\mathcal{S}_i^mk_{\ai}c_{\I}(\dv_i^m,\tau|\dv_{0_{\I}},t_{0_{\I}}) \mathrm{d}\tau.\label{eq:NI}
\end{align}
\end{theorem}
\begin{proof}
    We first characterize the number of insulin molecules absorbed by cell membrane $\parOi$. To this end, we define the change in the total number of insulin molecules due to molecule flux as \eqcite{Schulten2000LecturesBiophysics}{3.20} \begin{align}
    \dfrac{\mathrm{d}N_{\I}(\dv\in\parOi,t|\dv_{0_{\I}},t_{0_{\I}})}{\mathrm{d}t} = \int_{\parOi} D_{\I}\nabla_{\dv} c_{\I}(\dv,t|\dv_{0_{\I}},t_{0_{\I}}) \mathrm{d}\mathcal{S}.\label{eq:dNdt}
\end{align} The change in the total number of particles in \eqref{eq:dNdt} on the cell membrane is equivalent to the flux at the boundary specified in \eqref{eq:BCins}, i.e., we can substitute $D_{\I}\nabla_{\dv} c_{\I}(\dv,t|\dv_{0_{\I}},t_{0_{\I}})$ with $k_{\ai}c_{\I}(\dv_i^m,\tau|\dv_{0_{\I}},t_{0_{\I}})$. As we are only interested in the insulin molecules that result in the translocation of GLUT4, we can evaluate the time derivative during the time period from $t-\tau_{\G_2}$ to $t-\tau_{\G_3}$. This leads to \begin{align}
    N_{\I}(\dv\in&\parOi,t|\dv_{0_{\I}},t_{0_{\I}}) =\nonumber\\& \sum_{m=1}^{M_i}\int_{t-\tau_{\G_2}}^{t-\tau_{\G_3}} \mathcal{S}_i^mk_{\ai}c_{\I}(\dv_i^m,\tau|\dv_{0_{\I}},t_{0_{\I}})\mathrm{d}\tau.\label{eq:N}
\end{align} Then, the average number of insulin molecules received by the cells can be derived as \eqref{eq:NI}.
\end{proof}

\subsection{Glucose Subsystem Analysis\label{sec:glucoseSA}}
The reception of the glucose molecules by a cell occurs in both states $S_0$ and $S_1$. We first derive the impulse response of the glucose subsystem for both states in Sections~\ref{sec:anaS0} and \ref{sec:anaS1}, respectively. Then, we derive the glucose concentration in the environment and that detected by the CGM through convolution of those two states in Sections~\ref{sec:anaglu} and \ref{sec:anaCGM}, respectively.
\subsubsection{State $S_0$ (reflective cells)}\label{sec:anaS0}
In state $S_0$, insulin is absent and the PDE in \eqref{eq:cG} with the boundary conditions \eqref{eq:BC} for the blood vessel and \eqref{eq:BC-reflective} for the cells is valid. Using \thref{th:1}, we provide the glucose concentration for state $S_0$ in the following corollary.
\begin{corollary}\label{co:cGS0}
    The impulse response for the glucose subsystem in state $S_0$ is given by \begin{align}
        c^{S_0}_{\G}(\dv,t|\dv_{0_{\G}},t_{0_{\G}}) &=  \sum_{n=0}^{\infty} c_n(\dv,t|\dv_{0_{\G}},t_{0_{\G}}),\label{eq:cGS0}
    \end{align} where \begin{align}
            c_0(\dv&,t|\dv_{0_{\G}},t_{0_{\G}}) = \nonumber\\&\dfrac{1}{(4\pi D_{\G}(t-t_{0_{\G}}))^{3/2}}\exp{\bigg( -\dfrac{||\dv_{0_{\G}}-\dv||^2}{4D_{\G}(t-t_{0_{\G}})}\bigg)}, \label{eq:c0glucose}
        \end{align} and
        \begin{align}
        c_n(\dv&,t|\dv_{0_{\G}},t_{0_{\G}}) =\sum_{i=1}^{N_b}\sum_{m=1}^{M_i} c_{n-1}(\dv_i^m,t|\dv_{0_{\G}},t_{0_{\G}})\nonumber\\&*\dfrac{\mathcal{S}_i^m D_{\G} (\dv_i^m-\dv)\cdot\nv_i^m}{4D_{\G}t(4\pi D_{\G}t)^{3/2}}\exp{\bigg( -\dfrac{||\dv_i^m-\dv||^2}{4D_{\G}t}\bigg)}
        \nonumber\\&+ \sum_{m=1}^{M_i} c_{n-1}(\dv_{N_b}^m,t|\dv_{0_{\G}},t_{0_{\G}})\nonumber\\&*\dfrac{\mathcal{S}_{N_b}^m k_{\ai}}{(4\pi D_{\G}t)^{3/2}}\exp{\bigg( -\dfrac{||\dv_{N_b}^m-\dv||^2}{4D_{\G}t}\bigg)}. \label{eq:cnglucoseS0}
    \end{align}
\end{corollary}    

\subsubsection{State $S_1$ (fully absorbing cells)}\label{sec:anaS1}
In state $S_1$, insulin is present and the PDE in \eqref{eq:cG} with the boundary conditions \eqref{eq:BC} for the blood vessel and \eqref{eq:BC-absorbing} for the cells is valid. Using \thref{th:1}, we obtain the glucose concentration for state $S_1$ in the following corollary.
\begin{corollary}\label{co:cGS1}
The impulse response of glucose subsystem in state $S_1$ can be expressed as \begin{align}
        c^{S_1}_{\G}(\dv,t|\dv_{0_{\G}},t_{0_{\G}}) &=  \sum_{n=0}^{\infty} c_n(\dv,t|\dv_{0_{\G}},t_{0_{\G}}),\label{eq:cGS1}
    \end{align}  where $c_0(\dv,t|\dv_{0_{\G}},t_{0_{\G}})$ is given by \eqref{eq:c0glucose} and
    \begin{align}
        &c_n(\dv,t|\dv_{0_{\G}},t_{0_{\G}}) = \sum_{i=1}^{N_b-1}\sum_{m=1}^{M_i} \nabla_{\dv} c_{n-1}(\dv_i^m,t|\dv_{0_{\G}},t_{0_{\G}})\nonumber\\&\cdot\nv_i^m*\dfrac{\mathcal{S}_i^m D_{\G}}{(4\pi D_{\G}t)^{3/2}}\exp{\bigg( -\dfrac{||\dv_i^m-\dv||^2}{4D_{\G}t}\bigg)}
        \nonumber\\&+c_{n-1}(\dv_{N_b}^m,t|\dv_{0_{\G}},t_{0_{\G}})* \dfrac{\mathcal{S}_{N_b}^m}{(4\pi D_{\G}t)^{3/2}}\nonumber\\&\times\bigg\{k_{\mathrm{a}_{N_b}}+D_{\G} \dfrac{(\dv_{N_b}^m-\dv)\cdot\nv_{N_b}^m}{4D_{\G}t}\bigg\}\exp{\bigg( -\dfrac{||\dv_{N_b}^m-\dv||^2}{4D_{\G}t}\bigg)}.\label{eq:cnglucoseS1}
    \end{align}     
\end{corollary}  
 
\subsubsection{Glucose concentration} \label{sec:anaglu}
Since the glucose concentration derived for each state corresponds to an impulse response, the output concentration can be obtained via convolution of the input concentration and this impulse response. 
\begin{theorem}
    The glucose concentration can be derived as \eqref{eq:cGsolution}, shown at the top of next page, where $\int_\Omega(\cdot)\mathrm{d}\dv$ represents the volume integral over $\Omega$.
\end{theorem}
\begin{proof}
    When an impulsive input of insulin is introduced, the glucose subsystem, which is initially in state $S_0$, transitions to $S_1$ at time $t_{0_{\I}}+\tau_{\G_1}+\tau_{\G_2}$, and then reverts back to $S_0$ at time $t_{0_{\I}}+\tau_{\G_1}+\tau_{\G_2}+\tau_{\G_3}$, when the available insulin is depleted. Putting together all the impulse responses of each state, the solution to the PDE in \eqref{eq:cG} can be obtained as in \eqref{eq:cGsolution}.
\end{proof}

\begin{figure*}
\vspace{-.5cm}
    \begin{align}
    c_{\G}(\dv,t|\dv_{0_{\G}},t_{0_{\G}}) = \begin{cases}
        &F_{\G}(\dv_{0_{\G}},t_{0_{\G}})*c^{S_0}_{\G}(\dv,t|\dv_{0_{\G}},t_{0_{\G}}),\hfill\hspace{1cm} 0\leq t< t_{0_{\I}}+\tau_{\G_1}+\tau_{\G_2}\\
        &\int_{\Omega} \{F_{\G}(\dv_{0_{\G}},t_{0_{\G}})*c^{S_0}_{\G}(\dv',t_{0_{\I}}+\tau_{\G_1}+\tau_{\G_2}|\dv_{0_{\G}},t_{0_{\G}})\}*c^{S_1}_{\G}(\dv,t|\dv',t_{0_{\I}}+\tau_{\G_1}+\tau_{\G_2})\dr\dv',\hspace{1cm}\hfill\\&\hfill\hspace{2cm} t_{0_{\I}}+\tau_{\G_1}+\tau_{\G_2}\leq t< t_{0_{\I}}+\tau_{\G_1}+\tau_{\G_2}+\tau_{\G_3}\\
        &\int_{\Omega}\Big[\int_{\Omega} \{F_{\G}(\dv_{0_{\G}},t_{0_{\G}})*c^{S_0}_{\G}(\dv',t_{0_{\I}}+\tau_{\G_1}+\tau_{\G_2}|\dv_{0_{\G}},t_{0_{\G}})\}\\&\hspace{1cm}*c^{S_1}_{\G}(\dv'',t_{0_{\I}}+\tau_{\G_1}+\tau_{\G_2}+\tau_{\G_3}|\dv',t_{0_{\I}}+\tau_{\G_1}+\tau_{\G_2})\dr\dv'\Big]\\&\hspace{1cm}* c^{S_0}_{\G}(\dv,t|\dv'',t_{0_{\I}}+\tau_{\G_1}+\tau_{\G_2}+\tau_{\G_3})\dr\dv'',\hfill t\geq t_{0_{\I}}+\tau_{\G_1}+\tau_{\G_2}+\tau_{\G_3}
    \end{cases}\label{eq:cGsolution}
\end{align} 
\end{figure*}

\subsubsection{Detection of blood sugar level by the CGM}\label{sec:anaCGM}
The CGM detects the glucose concentrations within the subcutaneous layer. Modeling the CGM as a passive receiver, the glucose concentration detected by the CGM can be derived as
\begin{align}
    c_{\G}(\dv\in\parO_{\CGM},t|\dv_{0_{\G}},t_{0_{\G}})=\dfrac{\int_{\dv\in\parO_{\CGM}} c_{\G}(\dv,t|\dv_{0_{\G}},t_{0_{\G}}) \mathrm{d}\dv}{V_{\CGM}},\label{eq:CGM}
\end{align}where $\parO_{\CGM}$ is the set of all points on the CGM boundary and $V_{\CGM}$ is the volume of the CGM detection area.

\section{Insulin Injection Time Optimization\label{sec:opt}}
Effectively managing postprandial glucose levels is crucial for individuals with T1DM to maintain overall health and prevent both short-term and long-term complications. The management of postprandial glucose levels mainly focuses on regulating the peak value of glucose, typically accomplished by injecting insulin before meals. However, patients frequently miss these injections or fail to follow the prescribed timing. Therefore, we aim to identify the optimal window for insulin administration to ensure effective control of post-meal glucose levels, even when delays in injection occur. The ideal window for insulin injection can be mathematically represented as a specific range, which is defined by \begin{align}
    \tau_{0_{\I}} = \{t\in\mathbb{R}^+|0\leq t \leq \hat{t}_{0_{\I}}\},\label{eq:tauI0}
\end{align} where $\hat{t}_{0_{\I}}$ is an upper bound on the insulin injection time after a meal, and $t=0$ (i.e., the lower bound) is the time when the patient takes his/her meal.


Here, parameter $\hat{t}_{0_{\I}}$ is optimized to maximize the time interval between a meal and an insulin injection, while ensuring that post-meal glucose levels remain below a predefined threshold, $\phi$. Logically, the upper bound occurs when the peak glucose concentration is closest to the glucose threshold value. The distance between the peak glucose concentration and the threshold can be represented by the absolute difference $|c_{\pk}(t_{0_{\I}})-\phi|$. So, by minimizing the absolute difference, we can obtain the upper bound value. Thus, the corresponding optimization problem can be formulated as follows:\begin{mini}|s|{t_{0_{\I}}}{f(t_{0_{\I}})=|c_{\pk}(t_{0_{\I}})-\phi|}{\label{eq:opProb}}{}
    \addConstraint{t_{0_{\I}}}{\geq 0},
\end{mini} where $c_{\pk}(t_{0_{\I}})$ is the peak glucose concentration detected by the CGM when the insulin is injected at $t_{0_{\I}}$ and the optimal solution is exactly the upper bound $\hat{t}_{0_{\I}}$ of the insulin injection time. As the peak concentration in \eqref{eq:opProb} corresponds to the maximum concentration, the optimization problem can be reformulated as\begin{align}
    \min_{t_{0_{\I}}} |\max_{t} \{c_{\G}(\dv_{\CGM},t|\dv_{0_{\G}},0,t_{0_{\I}})\}-\phi|,\label{eq:minimax}
\end{align}where $\dv_{\CGM}$ is the midpoint of the CGM.\footnote{We only considered a single point as a receiver for simplicity of the subsequent analysis.} Therefore, the problem can be divided into determining the peak glucose concentration and subsequently deriving the optimal insulin injection time.

Deriving a closed-form expression for $\hat{t}_{0_{\I}}$ will allow to analyze the mathematical relationships and phenomena associated with T1DM. However, it is a great challenge to obtain a closed-form expression for $\hat{t}_{0_{\I}}$. As the blood vessel serves solely as a barrier and does not participate in the insulin-glucose mechanism, we first omit the impact of the blood vessel to reduce complexity and derive a closed-form expression for $\hat{t}_{0_{\I}}$ for this special case in \secref{sec:opcase}. Then, we numerically calculate $\hat{t}_{0_{\I}}$ using the gradient descent ascent (GDA) algorithm for the case when a blood vessel is present (i.e., the general case) in \secref{sec:generalOpins}.

\subsection{Special Case\label{sec:opcase}}
By disregarding the influence of the blood vessel (i.e., $k_{\ai}=0$)\footnote{Note that this simplification results in $c_n(\dv,t|\dv_{0_{\G}},t_{0_{\G}})=0$ for any $n\geq2$.} and then replacing the convolution with the approximation\footnote{As $f(t)*g(t) = \int_{0}^{t}f(\tau)g(t-\tau)d\tau$, we approximate the integral using the rectangular rule.} $f(t)*g(t) \approx tf(\frac{t}{2})g(\frac{t}{2})$ and incorporating the constraint $t_{0_{\G}}=0$, \eqref{eq:cnglucoseS0} becomes \begin{align}
    c_{\G}^{S_0}(&\dv,t|\dv_{0_{\G}},0) = \dfrac{1}{(4\pi D_{\G}t)^{3/2}}\exp{\bigg( -\dfrac{||\dv_{0_{\G}}-\dv||^2}{4D_{\G}t}\bigg)} \nonumber\\&+ \sum_{i=1}^{N_b}\sum_{m=1}^{M_i} \dfrac{\mathcal{S}_i^m(\dv_i^m-\dv)\cdot\nv_i^m}{16(\pi D_{\G}t)^{3}}\nonumber\\&\times\exp{\bigg( -\dfrac{||\dv_{0_{\G}}-\dv_i^m||^2+||\dv_i^m-\dv||^2}{2D_{\G}t}\bigg)},\label{eq:cglucoseS0reduced}
\end{align} and \eqref{eq:cnglucoseS1} becomes \begin{align}
    c_{\G}^{S_1}(&\dv,t|\dv_{0_{\G}},0) = \dfrac{1}{(4\pi D_{\G}t)^{3/2}}\exp{\bigg( -\dfrac{||\dv_{0_{\G}}-\dv||^2}{4D_{\G}t}\bigg)} \nonumber\\&- \sum_{i=1}^{N_b}\sum_{m=1}^{M_i} \dfrac{\mathcal{S}_i^m(\dv_{0_{\G}}-\dv_i^m)\cdot\nv_i^m}{16(\pi D_{\G}t)^{3}}\nonumber\\&\times\exp{\bigg( -\dfrac{||\dv_{0_{\G}}-\dv_i^m||^2+||\dv_i^m-\dv||^2}{2D_{\G}t}\bigg)}.\label{eq:cglucoseS1reduced}
\end{align}

\subsubsection{Peak Time and Concentration\label{sec:peaktc}}
We can derive an expression for the peak glucose concentration by focusing on the maximization term in \eqref{eq:minimax}. The examination of the reduced glucose concentration equations for state $S_0$ in \eqref{eq:cglucoseS0reduced} and state $S_1$ in \eqref{eq:cglucoseS1reduced} reveals an elevation in molecule concentration from the baseline during state $S_0$ and a decrease during state $S_1$. This observation implies that the peak time occurs during state $S_0$ and reduces the complexity of \eqref{eq:cGsolution} when seeking the peak concentration. In particular, we can restrict \eqref{eq:cGsolution} specifically for this analysis to \begin{align}
    c_{\G}(\dv,t|\dv_{0_{\G}},0) = F_{\G}(\dv_{0_{\G}},t_{0_{\G}})*c^{S_0}_{\G}(\dv,t|\dv_{0_{\G}},0). \label{eq:cGsolutionreduce}
\end{align} 

Considering a normalized impulsive glucose input, i.e., $F_{\G}(\dv_{0_{\G}},t_{0_{\G}})=\delta(\dv-\dv_{0_{\G}})\delta(t)$, we provide the approximate time $t_{\pk}$ at which the glucose concentration reaches its maximum in the following theorem. Note that to obtain a closed-form solution, we resort to approximating $\frac{a}{t^{n/2}}\exp(-\frac{b}{t})$, which reduces the equation from a higher-order equation to a solvable cubic equation. This leads to an analysis that is more tractable while capturing the essential behavior of the original function when $a\exp(-b)$ is the dominating term.
\begin{theorem}\label{th:3}
    The approximate time at which the glucose concentration reaches its maximum is obtained as \begin{align}
        t_{\pk} = \dfrac{t_{\pk}^{S_0}+t_{0_{\I}}+\tau-|t_{\pk}^{S_0}-(t_{0_{\I}}+\tau)|}{2},\label{eq:tpeak}
    \end{align} where $\tau$ is the approximate time when the system transitions from state $S_0$ to state $S_1$ after an insulin injection, i.e., $\tau = \tau_{\G_1}+\tau_{\G_2}$,\footnote{We note that $\tau_{\G_1}$ can be obtained numerically by setting \eqref{eq:NI} to one.} and $t_{\pk}^{S_0}$ is obtained as follows
    \begin{align}
        t_{\pk}^{S_0} = -\dfrac{1}{3A_{1}}\bigg(A_{2}+\xi\rho-\dfrac{\Delta_0}{\xi\rho}\bigg), \label{eq:tpeakS0}
    \end{align} with
    \addtocounter{equation}{-1}
    \begin{subequations}
        \begin{align}
        A_{1} =& -\dfrac{3}{16(\pi D_{\G})^{3/2}}\exp{\bigg( -\dfrac{||\dv_{0_{\G}}-\dv_{\CGM}||^2}{4D_{\G}}\bigg)},\end{align}
        \begin{align}
        A_{2} =& \dfrac{||\dv_{0_{\G}}-\dv_{\CGM}||^2}{32D_{\G}(\pi D_{\G})^{3/2}}\exp{\bigg( -\dfrac{||\dv_{0_{\G}}-\dv_{\CGM}||^2}{4D_{\G}}\bigg)},\end{align}
        \begin{align}
        \xi =& \dfrac{-1+\sqrt{3}}{2},\end{align}
        \begin{align}
        \rho =& \sqrt[3]{\dfrac{\Delta_0-\sqrt{\Delta_1^2-4\Delta_0^2}}{2}},\end{align}
        \begin{align}
        \Delta_0 =& A_{2}^2-3A_{1}A_{3},\end{align}
        \begin{align}
        \Delta_1 =& 2A_{2}^2-9A_{1}A_{3}+27A_{1}^2A_{4},\end{align}
        \begin{align}
        A_{3} =& \sum_{i=1}^{N_b}\sum_{m=1}^{M_i} -\dfrac{3\mathcal{S}_i^m(\dv_i^m-\dv_{\CGM})\cdot\nv_i^m}{32(\pi D_{\G})^3}\nonumber\\&\times\exp{\bigg(-\dfrac{||\dv_i^m-\dv_{0_{\G}}||^2+||\dv_i^m-\dv_{\CGM}||^2}{2D_{\G}}\bigg)},\end{align}\vspace{-0.5cm}
        \begin{align}
        A_{4} =& \sum_{i=1}^{N_b}\sum_{m=1}^{M_i} \dfrac{\mathcal{S}_i^m(\dv_i^m-\dv_{\CGM})\cdot\nv_i^m}{32D_{\G}(\pi D_{\G})^3}\nonumber\\&\times\Big(||\dv_i^m-\dv_{0_{\G}}||^2+||\dv_i^m-\dv_{\CGM}||^2\Big)\nonumber\\&\times\exp{\bigg(-\dfrac{||\dv_i^m-\dv_{0_{\G}}||^2+||\dv_i^m-\dv_{\CGM}||^2}{2D_{\G}}\bigg)}.
    \end{align}
    \end{subequations}
\end{theorem}
\begin{proof}
    See \apref{ap:3}.
\end{proof}

Finally, the objective function for the optimization problem in \eqref{eq:minimax} can be expressed as \begin{align}
        f(t_{0_{\I}})=|c_{\G}^{S_0}(\dv_{\CGM},t_{\pk}|\dv_{0_{\G}},0,t_{0_{\I}})-\phi|.\label{eq:ft0}
    \end{align}

\subsubsection{Optimal Insulin Injection Time\label{sec:optime}}
To determine the optimal insulin injection time, we convert the constrained problem in \eqref{eq:opProb} into an unconstrained problem by employing the method of Lagrange multipliers \cite{boyd2004convex} as\begin{align}
    \mathcal{L}(t_{0_{\I}},\lambda,s) = f(t_{0_{\I}})-\lambda(t_{0_{\I}} - s^2),
\end{align}where $\lambda$ is the Lagrange multiplier and $s^2$ is used to ensure the positivity of the equality constraint. To find the minimum of function $f(t_{0_{\I}})$, we set all partial derivatives to zero and arrive at \begin{subequations}
    \begin{align}
        \dfrac{\partial f(t_{0_{\I}})}{\partial t_{0_{\I}}} - \lambda &= 0,\label{eq:lag1}\\
        (t_{0_{\I}} - s^2) &= 0,\label{eq:lag2}\\
        2\lambda s &= 0.\label{eq:lag3}
    \end{align}
\end{subequations} Eq.~\eqref{eq:lag3} implies that two cases need to be considered: 
\begin{itemize}
    \item \textbf{Case 1:} When $s = 0$, \eqref{eq:lag2} indicates $t_{0_{\I}} = 0$. Substituting $t_{0_{\I}} = 0$ into function $f(t_{0_{\I}})$ satisfies the stationary condition but does not yield a minimum point.
    \item \textbf{Case 2:} When $\lambda = 0$, \eqref{eq:lag1} implies $\frac{\partial f(t_{0_{\I}})}{\partial t_{0_{\I}}} = 0$, and solving this equation leads to the minimum for the optimization problem.
\end{itemize}

\begin{theorem}\label{th:4}
    The approximate maximum insulin injection time after a meal with the postprandial glucose level below threshold $\phi$ is obtained as \begin{align}
        \hat{t}_{0_{\I}} = \begin{cases}
            \infty,\hfill c_{\G}^{S_0}(\dv_{\CGM},t_{\pk}^{S_0}|\dv_{0_{\G}},0)<\phi,\\
            t_{\pk}^{S_0}-\tau,\hspace{1cm}\hfill c_{\G}^{S_0}(\dv_{\CGM},t_{\pk}^{S_0}|\dv_{0_{\G}},0)=\phi,\\
            \dfrac{-B_{1}\pm\sqrt{(B_{1})^2-4B_{2}\phi}}{2B_{2}}-\tau, \\\hfill c_{\G}^{S_0}(\dv_{\CGM},t_{\pk}^{S_0}|\dv_{0_{\G}},0)>\phi.
        \end{cases},\label{eq:tI*}
    \end{align}where $\infty$ implies that the maximum insulin injection time is boundless\footnote{This further implies that the glucose level is already below the threshold and that no insulin injection is required.}, 
    \addtocounter{equation}{-1}\begin{subequations}
        \begin{align}
        B_{1} =& \dfrac{1}{(4\pi D_{\G})^{3/2}}\exp{\bigg( -\dfrac{||\dv_{0_{\G}}-\dv||^2}{4D_{\G}}\bigg)},\end{align}
        \begin{align}
        B_{2} =& \sum_{i=1}^{N_b}\sum_{m=1}^{M_i} \dfrac{\mathcal{S}_i^m(\dv_i^m-\dv)\cdot\nv_i^m}{16(\pi D_{\G})^{3}}\nonumber\\&\times\exp{\bigg( -\dfrac{||\dv_{0_{\G}}-\dv_i^m||^2+||\dv_i^m-\dv||^2}{2D_{\G}}\bigg)}.
    \end{align}
    \end{subequations}
\end{theorem}
\begin{proof}
    See \apref{ap:4}.
\end{proof}

Note that, similar to \thref{th:3}, we resort to approximating $\frac{a}{t^{n/2}}\exp(-\frac{b}{t})$, which reduces the equation to a solvable quadratic or cubic equation. This leads to an approximated analysis with an error that is negligible when $a\exp(-b)$ is the dominating term.

In \figref{fig:t1dmEq}, we illustrate the relationship between the subsystems and the optimization problem.
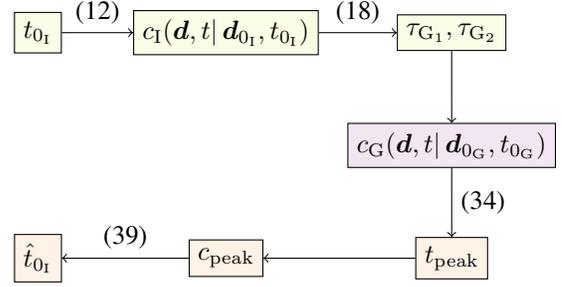
\begin{figure}
    \centering
    \tikzset{
        ->, 
        node distance=1.5cm,
        }
    \begin{tikzpicture}[scale=0.6]
        \node[fill=lime!10 , draw] (tI) {$t_{0_{\I}}$};
        \node[node distance=2.5cm,right of = tI,fill=lime!10 , draw] (cI) {$c_{\I}(\dv,t|\dv_{0_{\I}},t_{0_{\I}})$};
        \node[node distance=3cm,right of = cI,fill=lime!10 , draw] (t1) {$\tau_{\G_1}, \tau_{\G_2}$};

        \draw 
            (tI) edge[above] node{\eqref{eq:BIEsolution}} (cI)
            (cI) edge[above] node{\eqref{eq:NI}} (t1)
        ;
        
        \node[below of = t1,fill= violet!10 , draw] (cG) {$c_{\G}(\dv,t|\dv_{0_{\G}},t_{0_{\G}})$};

        \draw 
            (t1) edge[] node{} (cG)
        ;

        \node[below of = cG,fill= orange!10, draw] (tpeak) {$t_{\pk}$};
        \node[node distance=3cm,left of = tpeak,fill= orange!10, draw] (cpeak) {$c_{\pk}$};
        \node[node distance=2.5cm,left of = cpeak,fill= orange!10, draw] (tI*) {$\hat{t}_{0_{\I}}$};

        \draw 
            (cG) edge[right] node{\eqref{eq:tpeak}} (tpeak)
            (tpeak) edge[above] (cpeak)
            (cpeak) edge[above] node{\eqref{eq:tI*}} (tI*)
        ;
    \end{tikzpicture}
    \caption{Relationship of subsystem analyses in T1DM system and their relation to the considered optimization problem. The green, purple, and orange blocks originate from the insulin subsystem, the glucose subsystem, and the optimization problem, respectively.}
    \label{fig:t1dmEq}
\end{figure}
\subsection{General Case\label{sec:generalOpins}}
\begin{algorithm}[h!t]
 \caption{Gradient Descent Ascent for finding the optimal insulin injection time}
 \begin{algorithmic}[1]
 \renewcommand{\algorithmicrequire}{}
  \REQUIRE \textit{Calculation of the Peak Time and Concentration}
  \FOR {$t_{0_{\I}}\in T_{0_{\I}}$} 
    \STATE Initialization: $l_{0} \leftarrow 0.0001$
    \STATE $n \leftarrow 0$
    \REPEAT
        \STATE $l_{n+1}\leftarrow l_{n}+\gamma_{1}\nabla_{t} c_{\G}(\dv,l_{n}|\dv_{0_{\G}},t_{0_{\G}},t_{0_{\I}})$
        \STATE $n\leftarrow n+1$
    \UNTIL {$l_{n+1} < l_{n}$}
    \STATE $t_{\pk} \leftarrow l_{n}$
    \STATE $c_{\pk}(t_{0_{\I}})\leftarrow c_{\G}(\dv,t_{\pk}|\dv_{0_{\G}},t_{0_{\G}},t_{0_{\I}})$
  \ENDFOR
    \\\textit{Calculation of the Optimal Insulin Injection Time}
    \STATE Initialization: $m_{0} \leftarrow 0.0001$
    \STATE $n \leftarrow 0$
    \REPEAT
        \STATE $m_{n+1}\leftarrow m_{n}-\gamma_{2}\nabla_{t} f(m_{n})$
        \STATE $n\leftarrow n+1$
    \UNTIL{$m_{n+1} > m_{n}$}
    \STATE $\hat{t}_{0_{\I}} \leftarrow m_{n}$
 \RETURN $\hat{t}_{0_{\I}}$
 \end{algorithmic}
 \label{al:GDA}
 \end{algorithm}
Considering the intricate nature of the glucose concentration in \eqref{eq:cGsolution}, obtaining a closed-form solution to the optimization problem using conventional optimization techniques becomes challenging. Consequently, we resort to numerical methods to address the general case, which encompasses the presence of the blood vessel. Eq.~\eqref{eq:minimax} is in the form of a minimax optimization problem, which is frequently encountered in game theory and optimization\cite{boyd2004convex}. Such a problem involves two adversaries, each striving to minimize or maximize an objective function while anticipating the opponent's actions. In this scenario, the minimizing player typically employs gradient descent to iteratively refine his/her strategy, aiming to reduce his/her cost function. Simultaneously, the maximizing player employs gradient ascent to optimize his/her objective. Through continuous strategy updates based on their respective gradients, the players converge towards a Nash equilibrium, where neither can unilaterally improve their positions. This iterative process mirrors strategic interactions, proving crucial in determining optimal solutions within competitive environments. Given the efficacy of GDA in seeking equilibrium solutions within adversarial settings\cite{lin2020gradient}, we apply this algorithm to tackle \eqref{eq:minimax}. However, owing to the absolute value operation in \eqref{eq:minimax}, we segregate the gradient descent and ascent into two distinct parts rather than conducting them iteratively. \alref{al:GDA} outlines the procedure for calculating the optimal insulin injection time. \alref{al:GDA} first calculates the peak time of the glucose concentration for all possible insulin injection times in the set $T_{0_{\I}}$. In Line 2, the initial weight is initialized to a small non-zero value. Then a loop in Lines 4 to 7 that employs gradient ascent with $\gamma_1$ as the learning rate and $\nabla_t=\frac{\partial}{\partial t}$ as the Nabla operator with respect to time $t$ continues until the maximum weight is reached. This approach aims to identify the maximum of the glucose concentration. The resulting peak time is then used to derive the peak concentration (Line 9). We then employ gradient descent (Lines 13 to 16) with $\gamma_2$ as the learning rate to search for the upper limit of the optimal insulin injection time.

\section{Numerical Results\label{sec:numres}}
In this section, we employ agent-based simulations in Java to validate our T1DM analysis and investigate the impact of various diabetes management strategies on blood sugar levels. By treating each cell as an individual agent, agent-based simulation can be used to more accurately capture the multicellular dynamic compared to traditional equation-based simulations. We have chosen Java instead of MATLAB, as MATLAB is not capable of effectively executing the agent-based simulations specific to our case \cite{Matyjaszkiewicz2017BSimSimulator,jad2023CSK}. The simulation of molecule propagation is based on a particle-based method \cite{Cai2006ModellingMigration, Deng2016MolecularReceiver}, which involves storing the number of molecules and their respective locations. The displacement of molecules in each direction due to diffusion in a 3D fluid environment in one time step $t_s$ is calculated according to an independent Gaussian distribution with variance $2Dt_s$ and zero mean, i.e., $\{\mathcal{N}(0,2Dt_s),\allowbreak \mathcal{N}(0,2Dt_s),\allowbreak \mathcal{N}(0,2Dt_s)\}$, where $\mathcal{N}$ is the normal distribution function. Additionally, molecules that undergo degradation are removed from the environment with a probability of $k_\mathrm{d}t_s$ during the same time step. To account for the partially absorbing boundary condition in \eqref{eq:BC}, a collision with the boundary results in absorption with a probability $k_{\ai}\sqrt{{\pi t_s}/{D}}$ and reflection otherwise.  For the reflective boundary, molecules positioned inside the boundary area are reflected back into the environment. As for the fully absorbing boundary, molecules positioned inside the boundary area are counted and removed from the environment.

The simulation time interval $t_s$ is set as $0.01$s, and all results shown are averaged over $10,000$ realizations. The parameter values adopted for the T1DM system are summarized
in \tabref{tab:parameter}. Unless stated otherwise, we simulate ten cells in the unbounded environment with their centers located at $[0,2,0]$, $[2,2,2]$, $[2,3,3]$, $[4,2,4]$, $[5,4,1]$, $[1,1,3]$, $[1,3,4]$, $[3,1,4]$, $[6,3,1],$ and $[3,2,3]\mu$m. We set the radius of the cells, the radius of the blood vessel, and the length of the blood vessel as 0.5$\mu$m, 0.2$\mu$m, and 15$\mu$m, respectively\footnote{These values were chosen to balance model complexity and computational feasibility, acknowledging that blood vessel sizes vary widely (2 µm to 2 cm) and human cells typically range from 10–100 µm, depending on type and function\cite{scanlon2018essentials}.}. In all the figures, we use “Ana.”, “Sim.”, “Num.”, “S.C.”, and “G.C.” to abbreviate “Analytical”, “Simulation”, “Numerical”, “Special Case”, and “General Case”, respectively. 

\begin{table}[!th]
        \centering
        \caption{Parameter adopted for the considered T1DM system.}
        \label{tab:parameter}
        \begin{tabular}{|c|c|c|c|}
        \hline
             Parameter&Value&Unit &Ref. \\\hline
$D_{\I}$&150&$\mu\mathrm{m}^2/s$&\cite{shorten2007insulin}\\
$k_{\dr_{\I}}$&0.6&$\mathrm{h}^{-1}$&\cite{shorten2007insulin}\\
$D_{\G}$&670&$\mu\mathrm{m}^2/s$&\cite{francke2003phosphotransferase}\\
$k_{\ai}$, $1\leq i<N_b$&20&$\mu\mathrm{m}/s$&$^a$\\
$k_{\mathrm{a}_{N_b}}$&1&$\mu\mathrm{m}/s$&$^a$\\
$\tau_{\G_2}$&0.2&h&$^b$\\
$\tau_{\G_3}$&0.5&h&\cite{ferrannini2012physiology}\\
             \hline
        \end{tabular}
        
        \begin{flushleft}
    
        \footnotesize{$^a$The chosen absorption rate for blood vessels may not reflect practical scenarios; however, these values were chosen as reasonable approximations given the limited literature and variability based on patient conditions. An early peak in glucose absorption occurs approximately 30-40 minutes after ingestion \cite{ferrannini2012physiology}, leading to an absorption rate of 0.85 - 1.13 $\mu$m/s for blood vessels\cite{yu2016high, tordjman2017melatonin}. Cells, being the primary consumers of glucose \cite{el2010dynamic}, absorb it more efficiently through specialized transporter proteins, prompting us to use a higher absorption rate for cells compared to blood vessels.}\vspace{0.15cm}

        \footnotesize{$^b$GLUT4 translocation is a very fast process, and the maximum glucose uptake is achieved after around 10 minutes of insulin binding to its receptor \cite{fazakerley2022glut4}. However, there is still some variety depending on the patient; thus, for simplicity, we chose 0.2h for $\tau_{\G_2}$.}
        \end{flushleft}
\end{table}


\subsection{Validation of T1DM Analysis}  
\begin{figure}
    \centering
    \includegraphics[width=\linewidth]{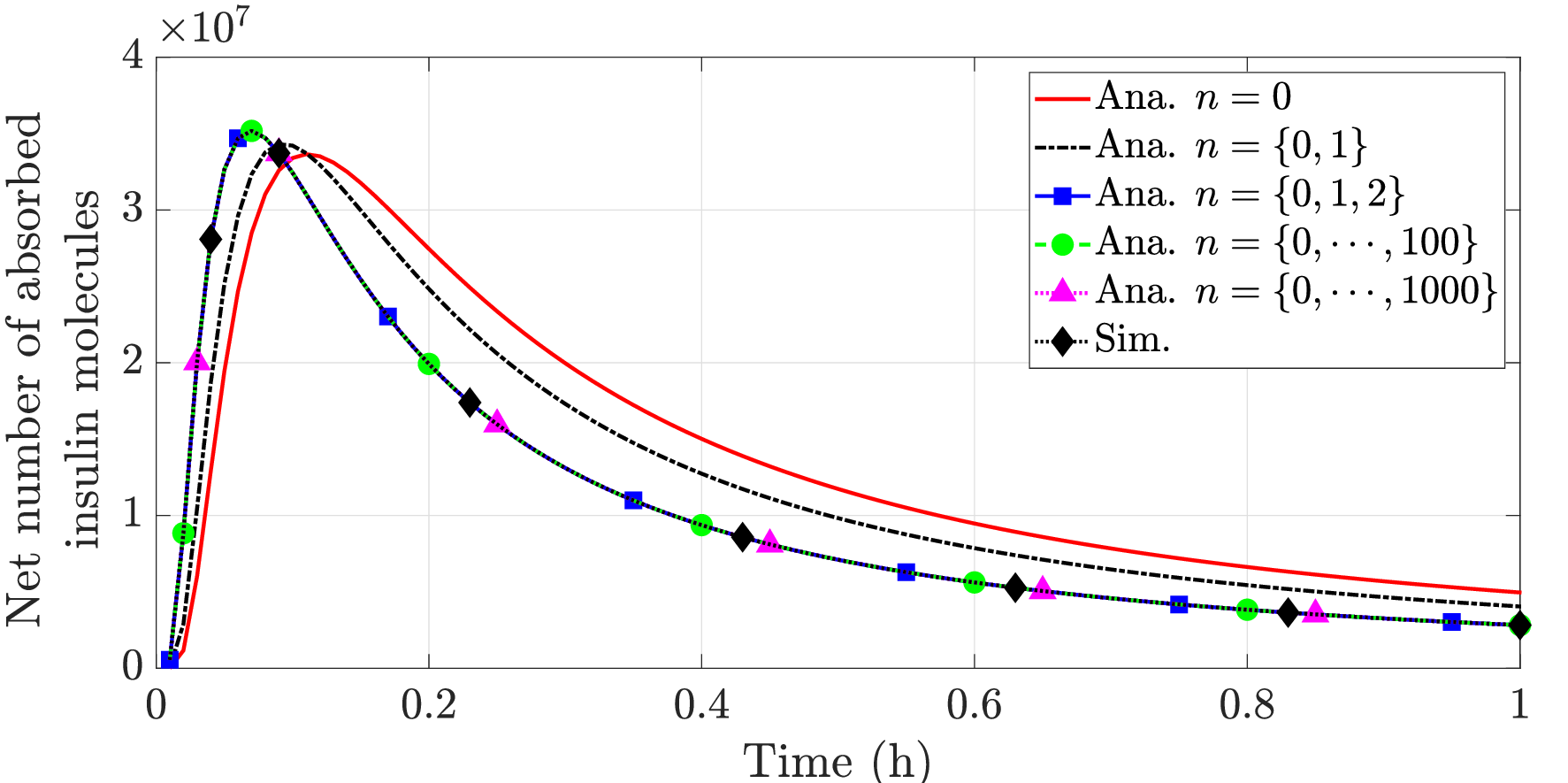}
    \caption{The convergence of the analytical result for the number of insulin molecules absorbed by a cell located at $\dv = [0,2,0]\mu$m for different $n$.}
    \label{fig:nvalue}
\end{figure}

\begin{figure}
    \centering
    \includegraphics[width=\linewidth]{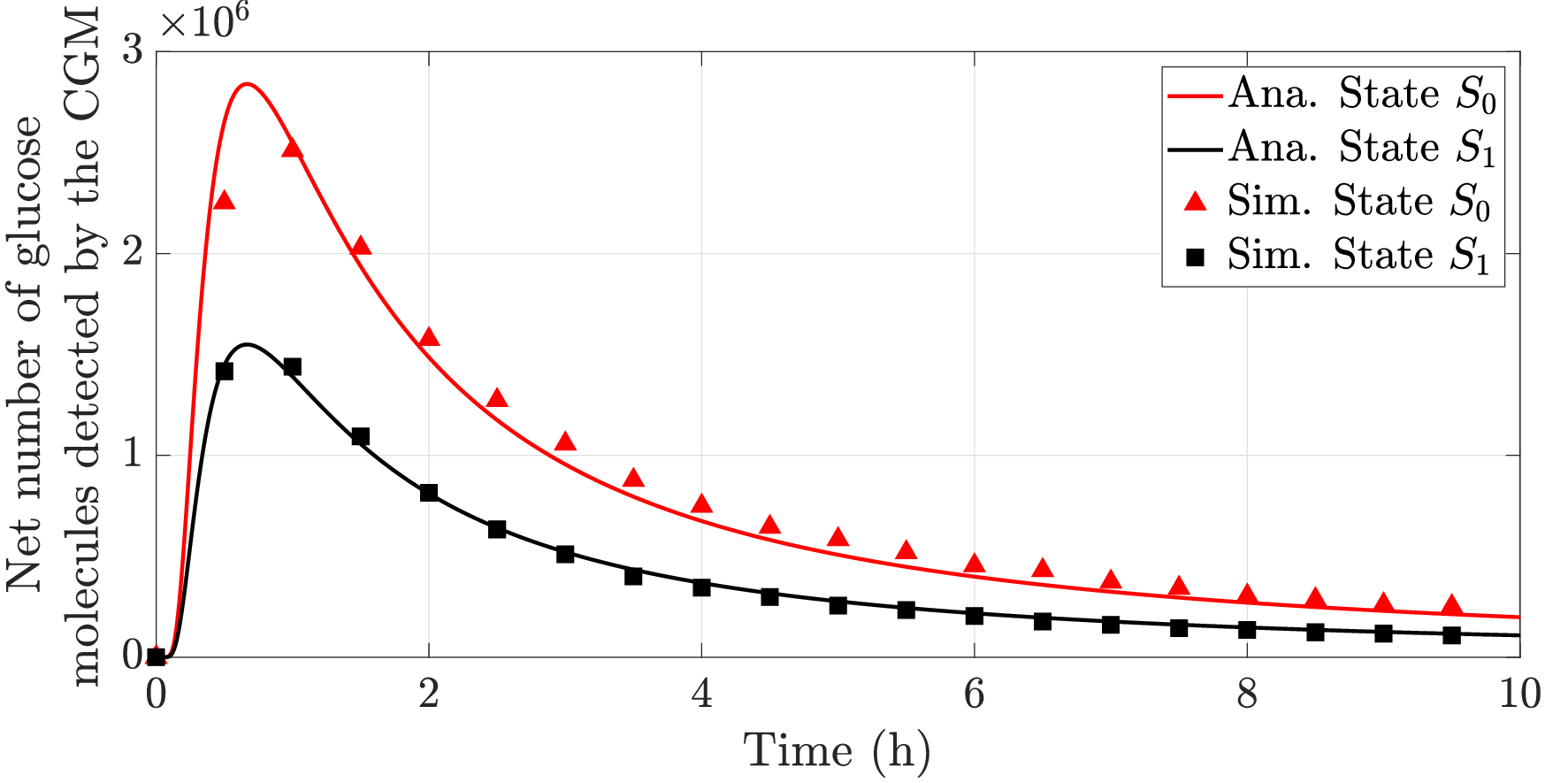}
    \caption{The net number of glucose molecules detected by the CGM for states $S_0$ and $S_1$.}
    \label{fig:glucose}
\end{figure}

\begin{figure}[t]
    \centering
    \includegraphics[width=\linewidth]{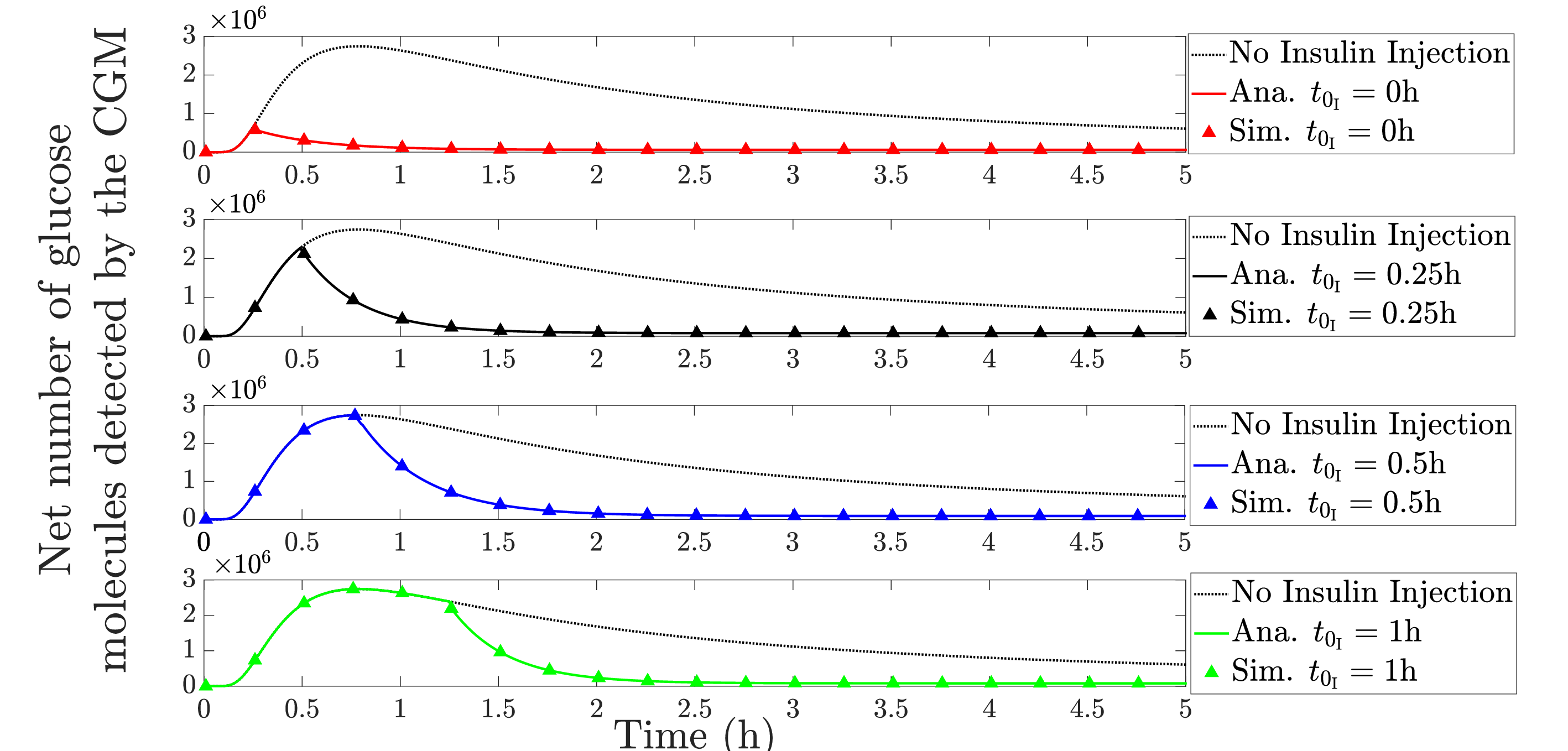}
    \caption{The net number of glucose molecules detected by the CGM for different insulin injection times.}
    \label{fig:T1DM}
\end{figure}

In Figs.~\ref{fig:nvalue}, \ref{fig:glucose}, and \ref{fig:T1DM}, we validate the analysis of the T1DM system presented in \secref{sec:recob}. In these figures, unless stated otherwise, $10^{10}$ insulin and glucose molecules are impulsively released at time $t=0$h at locations $\dv_{0_{\I}} = [0,0,0]\mu$m and $\dv_{0_{\G}} = [0,0,0]\mu$m, respectively. \figref{fig:nvalue} shows simulation results for the proposed insulin subsystem alongside the analytical net number of insulin molecules absorbed by a cell positioned at $\dv = [0,2,0]\mu$m computed using \coref{co:cI} and \eqref{eq:N}, with the upper limit in \eqref{eq:BIEsolution} truncated to the maximum $n$ shown in the legend. \figref{fig:glucose} compares the simulation results for the proposed glucose subsystem with the analytical results computed using \eqref{eq:cGS0} and \eqref{eq:cGS1}. \figref{fig:T1DM} shows the analytical net number of glucose molecules detected by the CGM located at $\dv_{\CGM}=[4,4,4]\mu$m using the result in \eqref{eq:CGM} for different insulin injection times. Our results indicate a close agreement between the analytical and simulation results and thus validate our theoretical analysis.

\figref{fig:nvalue} reveals that when the summation reaches $n=2$, the analytical curve converges. Thus, in the following simulations, $n$ is set as $\{0,1,2 \}$ to ensure the accuracy of the analysis and to avoid high computational complexity. \figref{fig:glucose} shows that a larger number of glucose molecules are received by CGM for state $S_0$  compared to state $S_1$. This difference arises from the fact that the reflective membrane of the cells in state $S_0$ results in the molecules being less likely to be removed through absorption, thereby increasing their likelihood of reaching the CGM.

In \figref{fig:T1DM}, we see that insulin injections lead to a decrease in the number of glucose molecules detected by the CGM compared to when insulin is not injected into the system. We also observe that insulin needs to be injected before the glucose reaches its peak to achieve effective blood sugar control, i.e., the glucose concentration does not exceed the threshold value at any time. We can further see that it takes time for the glucose level to start decreasing after the injection of insulin. This delay is caused by the delay in the reception of insulin and the GLUT4 translocation, which occurs approximately at $\tau_{\G_{1}}+\tau_{\G_{2}}$. 

\subsection{Optimization of Insulin Injection}
In this subsection, we first validate the analysis for optimal insulin injection provided in \secref{sec:opt} and then utilize it to examine the impact of different parameters on the objective function value $f(t_{0_{\I}})$  (i.e., the absolute value of the difference between the glucose peak and the threshold $\phi$) and the ideal insulin injection time window $\tau_{0_{\I}}$. We set $\phi = 3.325$mM as the recommended sugar level is advised to remain below 140mg/dL.\footnote{The typical basal blood sugar level rests at 70-90mg/dL and the recommended maximum sugar level is advised to be stopped below 140mg/dL\cite{zhou2009reference,park2013fasting}. Thus, the rise in blood sugar level after a meal should not exceed 60mg/dL, i.e., $\phi \leq 3.325$mM.}

\figref{fig:Optimization} shows simulation results, the analytical results from \eqref{eq:ft0} and \eqref{eq:tI*}, and the numerical results obtained with \alref{al:GDA} to validate and examine the impact of the special and general cases presented respectively in Sections~\ref{sec:opcase} and \ref{sec:generalOpins} on the objective function value and the ideal insulin injection time window. The close agreement between analytical and simulation results for the special case as well as between numerical and simulation results for the general case, validates our analytical approach and affirms the effectiveness of the algorithm. 

We further observe that the upper bound of insulin injection time $\hat{t}_{0_{\I}}$ is achieved when the objective function value is minimized, i.e., $f(t_{0_{\I}})=0$. To illustrate this, we plot the ideal insulin injection time window $\tau_{0_{\I}}$ by using \eqref{eq:tauI0} and the computed upper bound of insulin injection time $\hat{t}_{0_{\I}}$ (shown by the vertical bars). When the blood vessel is take into account (general case), the insulin injection time window grows. This occurs due to the reduced glucose availability in the system, which requires additional time for the peak concentration to exceed the threshold.

While we observe a close agreement of $\tau_{0_{\I}}$ for the general case, we also observe that there is a small error between the simulated $\tau_{0_{\I}}$ and the analytical $\tau_{0_{\I}}$ for the special case. This is mainly due to the fact that the closed-form solutions shown in \eqref{eq:ft0} and \eqref{eq:tI*} were obtained based on multiple simplifications of the environment and approximations of the analysis, consequently introducing errors.

In Table III, we further investigate the impact of the approximations in \textbf{Theorems~\ref{th:3}} and \textbf{\ref{th:4}} on the accuracy of the upper bound $\hat{t}_{0_{\I}}$. \tabref{tab:opError} shows $\hat{t}_{0_{\I}}$ and the computed percentage errors for the special and general cases with respect to the simulated result. We observe a close agreement of upper bounds for the numerical results (i.e., \alref{al:GDA}) and simulations. We further notice that the percentage error increases with the number of cells. But, we also observe that the algorithm results seem to be less influenced by the number of cells compared to the analytical results, which further confirms the impact of the simplification made to arrive at the analytical results.


\begin{figure}[t]
    \centering
    \includegraphics[width=\linewidth]{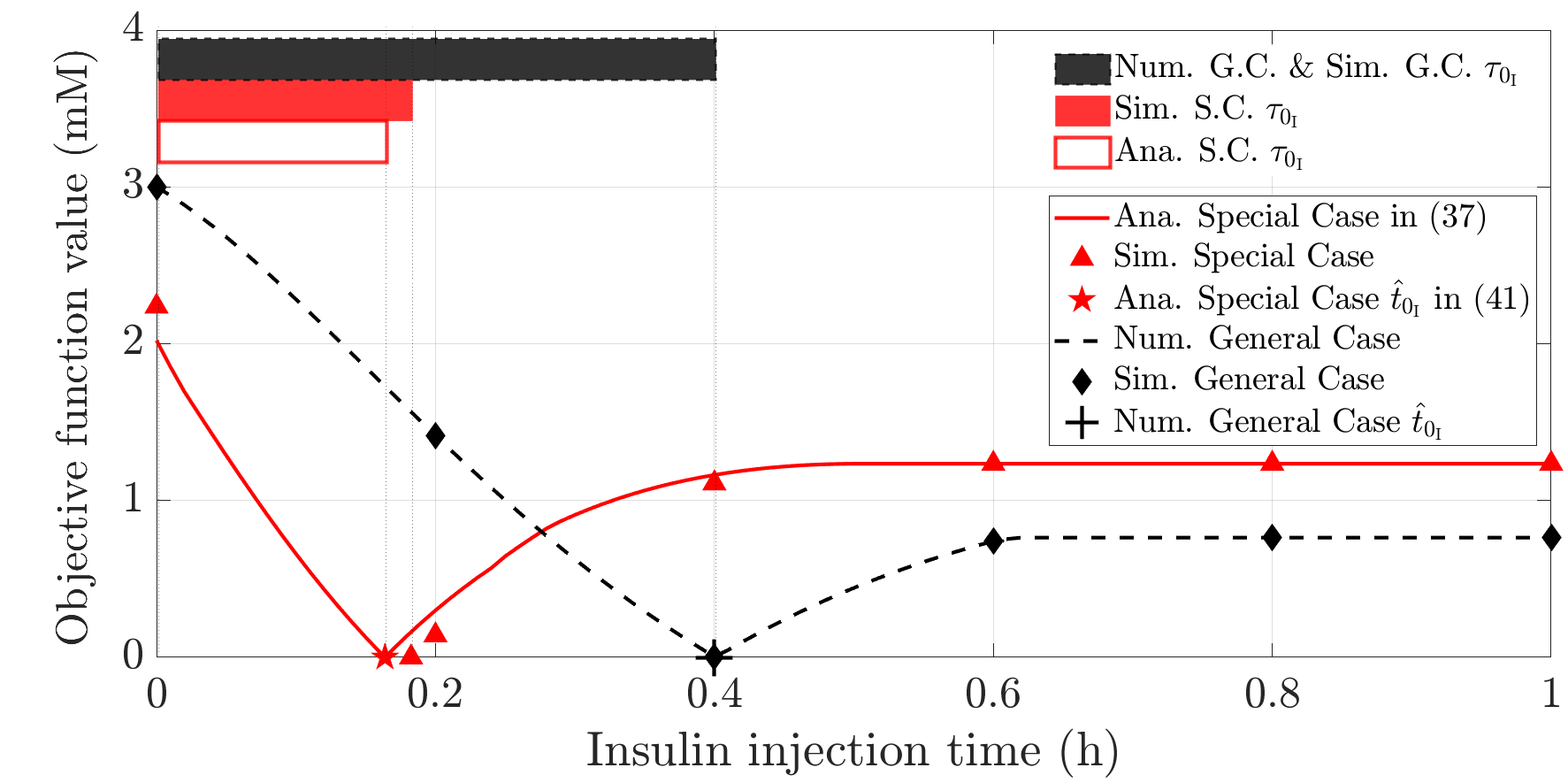}
    \caption{The objective function value and the corresponding ideal insulin injection time window (vertical bars) for special and general cases.}
    \label{fig:Optimization}
\end{figure}

\begin{table}
    \centering
    \caption{The accuracy of analytical and numerical optimization compared to particle-based simulation.}
    \begin{tabular}{c|c|c|c|c|c|c}
    \hline
     No.    & \multicolumn{3}{c|}{Special Case} & \multicolumn{3}{c}{General Case} \\\cline{2-7}
      of   & Sim.      & Ana.     & Error     & Sim.      & Num.     & Error     \\
        cells & (h)      & (h)     &\%     & (h)      & (h)     & \%     \\\hline\hline
2& 0.163    & 0.164   & 0.61    & 0.366    &0.367    &0.27          \\
10& 0.183    & 0.164   & 10.4    &  0.400   &0.401    &0.25          \\\hline
\end{tabular}
    \label{tab:opError}
\end{table}


In \figref{fig:DrugIns} and \figref{fig:MealIns}, we utilize the proposed optimization framework to examine the impact of different parameters on glucose regulation.

\subsubsection{Insulin Dispersion}\label{sec:DrugIns}
\begin{figure}[t]
    \centering
    \includegraphics[width=\linewidth]{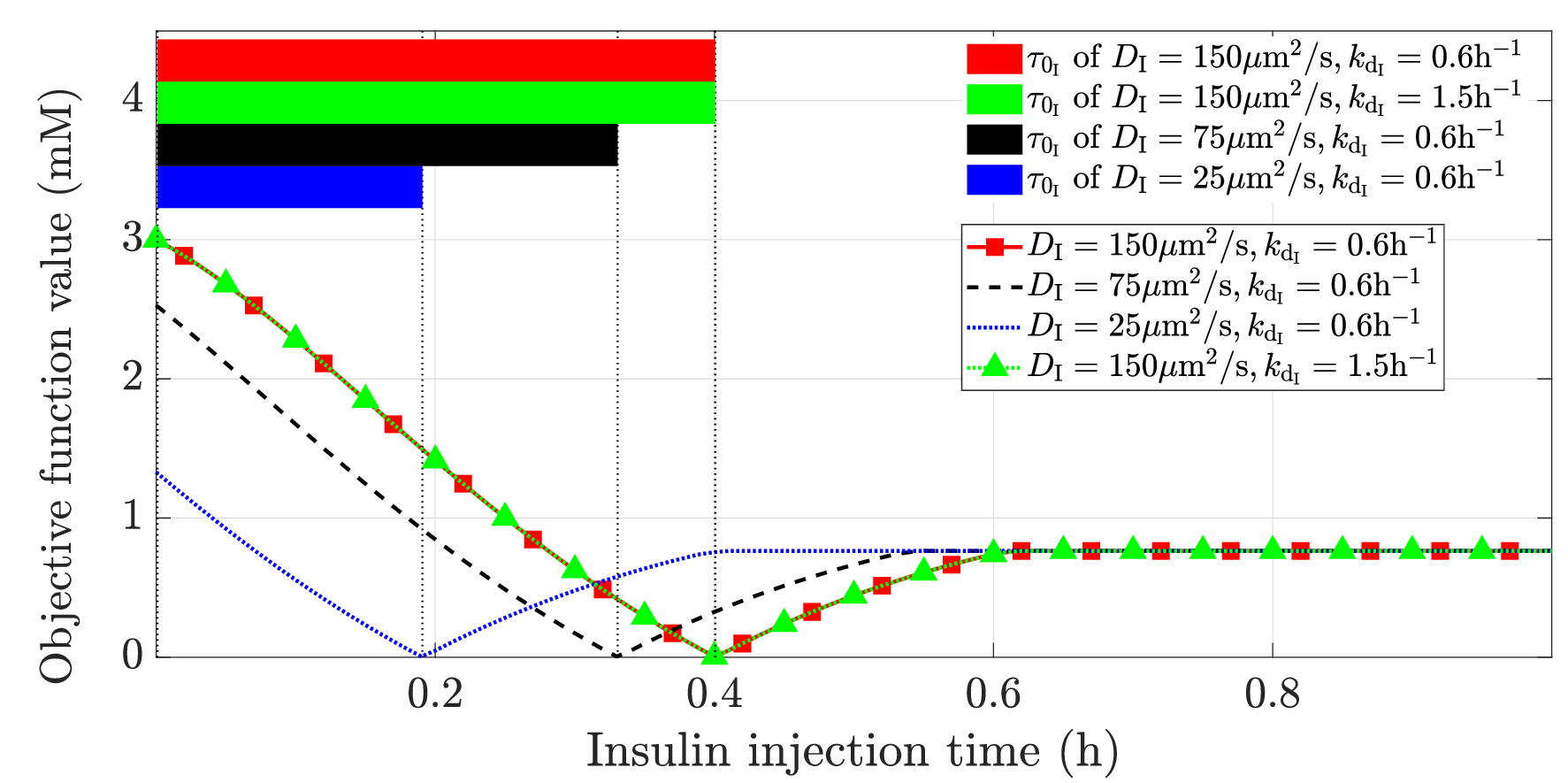}
    \caption{The impact of insulin dispersion on the objective function value and the corresponding ideal insulin injection time window.}
    \label{fig:DrugIns}
\end{figure}

Exploiting \alref{al:GDA}, \figref{fig:DrugIns} plots the ideal time window for insulin injection for different diffusion coefficients and degradation rates of insulin, which can correspond to different types of insulin, e.g., rapid-acting, short-acting, and long-acting insulin \cite{modi2007diabetes}. It is evident that a lower diffusion rate shortens the insulin injection window after an impulsive glucose input. This is due to the fact that insulin with a lower diffusion rate takes a longer time to spread and reach the cell. Therefore, a higher diffusion rate is preferred for impulsive glucose input to give patients more flexibility in injecting insulin after their meals. Interestingly, the degradation rate of insulin is not notably affected by blood sugar control. This is possibly due to the fact that the GLUT4 transporter remains functional on the cell membrane longer than the duration for which the glucose molecules are available in the environment, i.e., the glucose has been used up before the effect of insulin degradation becomes relevant. 

\subsubsection{Glucose Emission}\label{sec:MealIns}
\begin{figure}
    \centering
    \includegraphics[width=\linewidth]{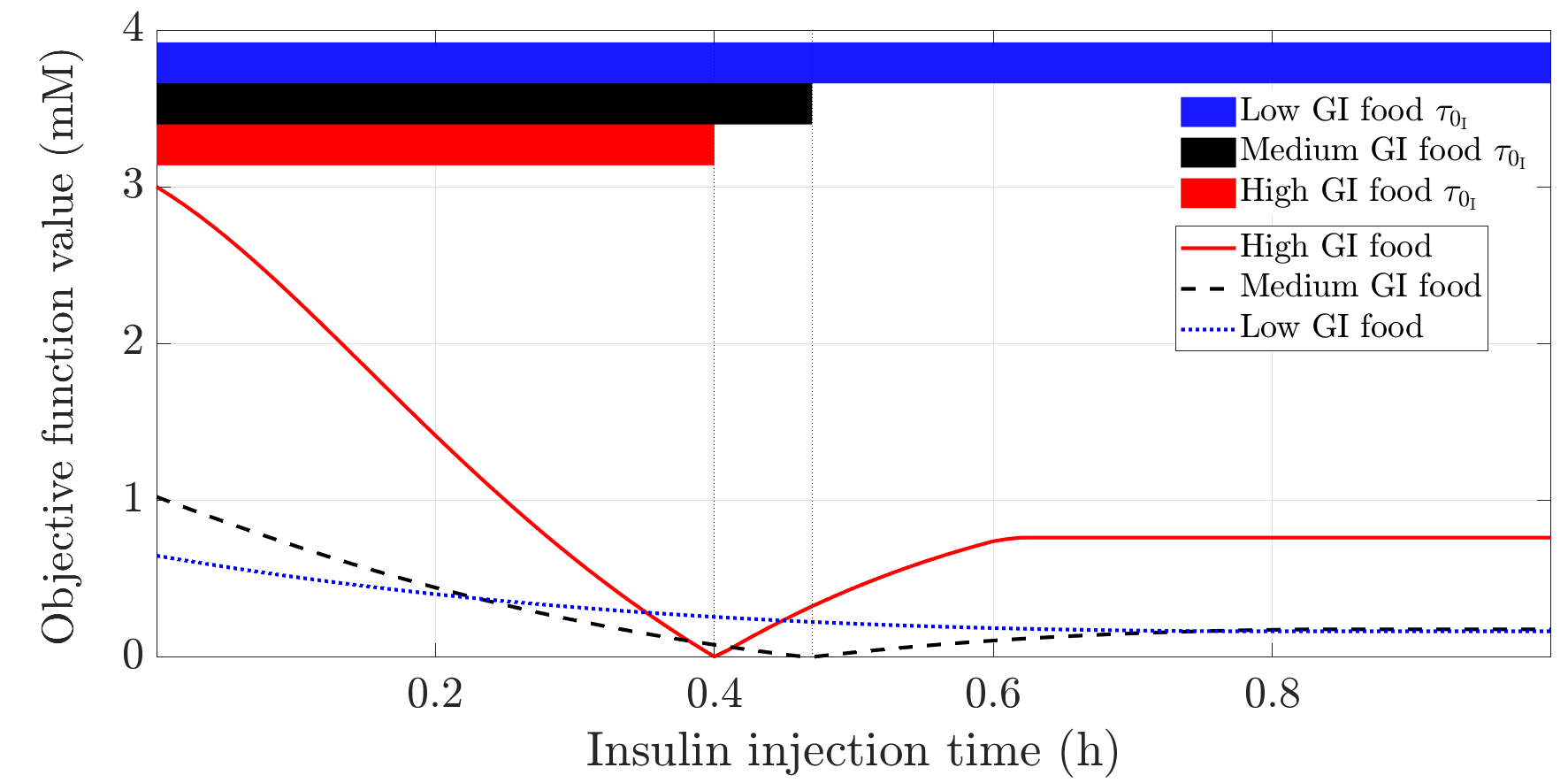}
    \caption{The impact of the glucose emission profile on the objective function value and the corresponding ideal insulin injection time window.}
    \label{fig:MealIns}
\end{figure}

\figref{fig:MealIns} examines the impact of glucose emission into the subcutaneous layer on the ideal window for insulin injection.\footnote{Glucose emission into the subcutaneous layer is linked to the glycemic-index (GI) of the food taken. Prioritizing low GI foods and adhering to dietary guidelines supports stable blood sugar levels and sustained energy throughout the day, especially when combined with regular exercise, aiding effective diabetes management \cite{gonzalez2023evolutionary}.} For simplicity, we model the glucose emission into the subcutaneous layer from a meal as \begin{align}
    F_{\G}(d_{0_{\G}},t) = A\rect\left(\frac{t}{W}-0.5\right)\delta(\dv-\dv_{0_{\G}}),
    \end{align}
where $A$ represents the amplitude of glucose concentration released, $\rect(\cdot)$ denotes the rectangular function, and $W$ determines the duration of glucose release into the system. 

\figref{fig:MealIns} shows numerical results obtained from \alref{al:GDA} for three different meal profiles having the same number of glucose molecules: high GI food ($A = 16.6$M, $W \rightarrow 0$)\footnote{As $W\rightarrow 0$, the rectangular function is converted into a direct delta function \cite{khare2015fourier}.}, medium GI food ($A = 8.30$M, $W =0.5$), and low GI food ($A = 4.15$M, $W =1$). We observe that the higher the GI of the food the patient eats, the shorter the insulin injection time window. We further observe that the objective function value reaches saturation at around 0.8h when eating low GI food; hence, the minimum objective function value occurs at $t_{0_{\I}}>0.8$h. Therefore, the upper bound for the insulin injection time does not exist for low GI food, i.e., $\hat{t}_{0_{\I}}\rightarrow\infty$. This suggests that by eating low GI food, the patient can potentially prevent the need for insulin, which confirms that proper meal management can be very effective in treating diabetes.


\section{Conclusion\label{sec:concl}}
In this paper, we proposed a model for the insulin-glucose interaction in the subcutaneous layer of a T1DM patient and modeled it as a multicellular MC system. For simplicity of the analysis, we established an analytical framework that divides the T1DM system into insulin and glucose subsystems. We then used Green’s second identity to reformulate the considered problem into a BIE and through the use of the Adomian decomposition method, the insulin and glucose concentrations in the T1DM system were derived. Simulation results obtained from an agent-based simulator showed a close match between analytical and simulation results; thus, confirming the validity of our theoretical analysis. We then utilized optimization methods to derive the most effective window for injecting insulin to manage the postprandial glucose levels of T1DM patients. Our optimization also allowed us to investigate the impact of different types of insulin and meal management strategies on glucose regulation. 

Our future work involves not only refining the mathematical framework by incorporating the insulin dissociation kinetics, basal parameters, and the insulin reception mechanism inside a cell, but also addressing physiological variability, dietary complexities, and long-term implications to enhance the model's real-world relevance. Thus, by translating these future findings into practical applications, i.e., developing personalized treatment protocols or integrating automated insulin delivery systems, our work has the potential to bridge the gap between theoretical modeling and clinical implementation.


\appendices
\allowdisplaybreaks

\section{Proof of \thref{th:BIEsolution}}\label{ap:2}
 We use the Adomian decomposition method \cite{rahman2007integral} and express $C_{\I}(\dv',\omega|\dv_{0_{\I}},t_{0_{\I}})$ in the form of a series as \begin{align}
        C_{\I}(\dv',\omega|\dv_{0_{\I}},t_{0_{\I}}) = \sum_{n=0}^{\infty} C_n(\dv',\omega|\dv_{0_{\I}},t_{0_{\I}}),\label{eq:ADM1}
    \end{align} where $C_0(\dv',\omega|\dv_{0_{\I}},t_{0_{\I}})$ is the term outside the integral \eqref{eq:generalBIE}, i.e., $C_0(\dv',\omega|\dv_{0_{\I}},t_{0_{\I}}) = H_{\I}(\dv_{0_{\I}},\omega|\dv')e^{-j\omega t_{0_{\I}}}$. Substituting \eqref{eq:ADM1} into \eqref{eq:generalBIE}, we can express the BIE in the form of a series as follows\begin{align}
        &\sum_{n=0}^{\infty}C_n(\dv',\omega|\dv_{0_{\I}},t_{0_{\I}})=\sum_{i=N_b^R + 1}^{N_b}\oint_{\parOi} H_{\I}(\dv,\omega|\dv') \nonumber\\&\nabla_{\dv} \bigg\{\sum_{n=0}^{\infty} C_n(\dv,\omega|\dv_{0_{\I}},t_{0_{\I}})\bigg\}\cdot\nv_i\mathrm{d}\mathcal{S} + H_{\I}(\dv_{0_{\I}},\omega|\dv')e^{-j\omega t_{0_{\I}}} \nonumber\\&+\sum_{i=1}^{N_b^R}\oint_{\parOi} \bigg\{\sum_{n=0}^{\infty} C_n(\dv,\omega|\dv_{0_{\I}},t_{0_{\I}})\bigg\}[k_{\ai}H_{\I}(\dv,\omega|\dv')\nonumber\\&-D\nabla_{\dv} H_{\I}(\dv,\omega|\dv')\cdot\nv_i] \mathrm{d}\mathcal{S}.\label{eq:ADM-BIE}
    \end{align}

    Assuming that the sum converges absolutely to $C_{\I}(\dv',\omega|\dv_{0_{\I}},t_{0_{\I}})$, we can interchange the summation and the integral in \eqref{eq:ADM-BIE}. Then, by expanding the summation on both sides and associating each $C_n(\dv',\omega|\dv_{0_{\I}},t_{0_{\I}})$ in a recursive manner, we obtain $C_{n}(\dv',\omega|\dv_{0_{\I}},t_{0_{\I}})$ as shown in \eqref{eq:ADM-BIE4} on the top of the next page.
    
    To obtain the solution in the time domain, we take the inverse Fourier transform of \eqref{eq:ADM1} and \eqref{eq:ADM-BIE4}. Then, by applying the meshing and rectangular rule \cite{zoofaghari2021semi} to resolve the surface integral, the solution in the time domain is given by \eqref{eq:BIEsolution}. Lastly, by replacing $\dv'$ with $\dv$, we arrive at \eqref{eq:BIEsolution}-\eqref{eq:cngeneral}.
    \begin{figure*}[!t]
    \begin{align}
        \begin{dcases}
            C_0(\dv',\omega|\dv_{0_{\I}},t_{0_{\I}}) =& H_{\I}(\dv_{0_{\I}},\omega|\dv')e^{-j\omega t_{0_{\I}}} \\
            C_n(\dv',\omega|\dv_{0_{\I}},t_{0_{\I}}) =& \sum_{i=1}^{N_b^R}\oint_{\parOi} C_{n-1}(\dv',\omega|\dv_{0_{\I}},t_{0_{\I}})[k_{\ai}H_{\I}(\dv,\omega|\dv')-D\nabla_{\dv} H_{\I}(\dv,\omega|\dv')\cdot\nv_i] \mathrm{d}\mathcal{S} \\&+\sum_{i=N_b^R + 1}^{N_b}\oint_{\parOi} H_{\I}(\dv,\omega|\dv') \nabla_{\dv}  C_{n-1}(\dv',\omega|\dv_{0_{\I}},t_{0_{\I}})\cdot\nv_i\mathrm{d}\mathcal{S}, \;\;\; i\geq 1
        \end{dcases}\label{eq:ADM-BIE4}
    \end{align}
    \vspace{-.5cm}
    \end{figure*} 

\section{Proof of \thref{th:3}}\label{ap:3}
The peak time for the glucose concentration either appears at the time that the system transits to state $S_1$ or at the peak time of state $S_0$, i.e., $\min{(t_{\pk}^{S_0},t_{0_{\I}}+\tau)}$. This can be written mathematically as in \eqref{eq:tpeak}.

Then, we can obtain $t_{\pk}^{S_0}$ by solving the maximization by equating the first-order derivative of \eqref{eq:cglucoseS0reduced} to zero. To obtain a closed-form solution, we resort to approximating $\frac{a}{t^{n/2}}\exp(-\frac{b}{t})$ as $\frac{a}{t^{\lceil 0.7n\rceil}}\exp(-b)$. Due to the exponential decay and term $b$ being larger than $t$, we can take the dominate term as the approximation, i.e., $\exp(-\frac{b}{t})$ can be approximated as $\exp(-b)$. Then, to compensate for the term $\exp(\frac{1}{t})$ and allow for factoring into a cubic equation, the term $\frac{a}{t^{n/2}}$ is mapped onto $\frac{a}{t^{\lceil 0.7n\rceil}}$. As such, $\frac{a}{t^{5/2}}\exp{(-\frac{b}{t})}$, $\frac{a}{t^{7/2}}\exp{(-\frac{b}{t})}$, $\frac{a}{t^{4}}\exp{(-\frac{b}{t})}$, and $\frac{a}{t^{5}}\exp{(-\frac{b}{t})}$ is approximated as $\frac{a\exp{(-b)}}{t^{4/3}}$, $\frac{a\exp{(-b)}}{t^{5/3}}$, $\frac{a\exp{(-b)}}{t^{6/3}}$, and $\frac{a\exp{(-b)}}{t^{7/3}}$, respectively, which leads to \begin{align}
    \dfrac{\mathrm{d}}{\mathrm{d}t_{0_{\I}}}c_{\G}^{S_0}(\dv,t|\dv_{0_{\G}},0) &\approx \dfrac{A_{1}t^{3/3}+A_{2}t^{2/3}+A_{3}t^{1/3}-A_{4}}{t^{7/3}}.\label{eq:ddtIcGs0}
\end{align} Note that the error of this approximation is negligible in the case that $a\exp{(-b)}$ is the dominating term. Setting \eqref{eq:ddtIcGs0} to zero is equivalent to solving the cubic equation $A_{1}t^3+A_{2}t^2+A_{3}t-A_{4}=0$. Consequently, solving for $t$ in this equation will yield $t_{\pk}^{S_0}$ as given in \eqref{eq:tpeakS0}.

\section{Proof of \thref{th:4}}\label{ap:4}
\textbf{Case 2} implies that the solution to the optimization problem in \eqref{eq:opProb} is found by solving $\frac{\partial f(t_{0_{\I}})}{\partial t_{0_{\I}}} = 0$. However, $f(t_{0_{\I}})$ is non-convex due to $t_{\pk}=t_{\pk}^{S_0}$ when $t_{\pk}^{S_0} \leq t_{0_{\I}}+\tau$, resulting in a saddle point when $t_{0_{\I}} \geq t_{\pk}^{S_0} - \tau$. Based on this, the following inferences can be made:\begin{itemize}
    \item When $c_{\G}^{S_0}(\dv_{\CGM},t_{\pk}^{S_0}|\dv_{0_{\G}},0)<\phi$, the glucose concentration observed by the CGM consistently remains below the threshold, rendering insulin injection unnecessary.
    \item When $c_{\G}^{S_0}(\dv_{\CGM},t_{\pk}^{S_0}|\dv_{0_{\G}},0)=\phi$, the peak concentration time represents the maximum allowable time for insulin injection before the glucose concentration surpasses the threshold. Thus, the optimal injection time becomes $\tau_{0_{\I}}=t_{\pk}^{S_0}-\tau$.
    \item When $c_{\G}^{S_0}(\dv_{\CGM},t_{\pk}^{S_0}|\dv_{0_{\G}},0)>\phi$, the minimum occurs at $t_{\pk} = t_{0_{\I}}+\tau$.
\end{itemize}
For the last case, we determine $\frac{\mathrm{d}}{\mathrm{d}t_{0_{\I}}}|c_{\pk}(t_{0_{\I}}+\tau)-\phi|=0$ to solve the optimization problem. Employing the property of absolute value, $\frac{\mathrm{d}}{\mathrm{d}t_{0_{\I}}}|c_{\pk}(t_{0_{\I}}+\tau)-\phi|$ can be expressed as ${\frac{\mathrm{d}}{\mathrm{d}t_{0_{\I}}}c_{\pk}(t_{0_{\I}}+\tau)}{c_{\pk}(t_{0_{\I}}+\tau)-\phi}$. Introducing variable $u=t_{0_{\I}}+\tau$ reduces the complexity, resulting in having to solve ${\frac{\mathrm{d}}{\mathrm{d}u}c_{\pk}(u+\tau)}{c_{\pk}(u+\tau)-\phi}=0$.

Similar to \thref{th:3}, we approximate not only $\frac{a}{t^{5/2}}\exp{(-\frac{b}{t})}$, $\frac{a}{t^{7/2}}\exp{(-\frac{b}{t})}$, $\frac{a}{t^{4}}\exp{(-\frac{b}{t})}$, and $\frac{a}{t^{5}}\exp{(-\frac{b}{t})}$, but also $\frac{a}{t^{3/2}}\exp{(-\frac{b}{t})}$ and $\frac{a}{t^{3}}\exp{(-\frac{b}{t})}$ as $\frac{a\exp{(-b)}}{t}$ and $\frac{a\exp{(-b)}}{t^{2}}$, respectively. Consequently, this leads to rewriting the problem as $(A_{1}u+A_{2}u^{2/3}+A_3u^{1/3}+A_4)(B_{1}u+B_{2}u^{2}-\phi)=0.$ Solving the cubic part of the equation results in $u=t_{\pk}$, which, after manipulation, yields the solution for the case when $c_{\G}^{S_0}(\dv_{\CGM},t_{\pk}^{S_0}|\dv_{0_{\G}},0)=\phi$. Subsequently, solving the quadratic part of the equation leads to the final case in \eqref{eq:tI*}.

\singlespacing
\bibliographystyle{IEEEtran}

\end{document}